\documentclass[11pt]{amsart}
\usepackage{geometry}
\usepackage{graphicx}
\usepackage{amssymb}
\usepackage{epstopdf}
\usepackage{amsmath,amscd}
\usepackage{amsthm}
\usepackage{enumerate}
\usepackage{url,verbatim}
\usepackage{subfig}
\usepackage{enumitem}

\usepackage[normalem]{ulem}
\usepackage{fixme}

\RequirePackage[colorlinks,citecolor=blue,urlcolor=blue]{hyperref}
\usepackage{breakurl}
\theoremstyle{plain}
\newtheorem{theorem}{Theorem}
\newtheorem{definition}[theorem]{Definition}
\newtheorem{lemma}[theorem]{Lemma}

\newtheorem{example}[theorem]{Example}
\newtheorem{assumption}[theorem]{Assumption}

\newtheorem{claim}[theorem]{Claim}

\newcommand\ol{\overline}

\newcommand\sB{{\mathcal B}}

\newcommand\RR{{\mathbb R}}

\newcommand\NN{{\mathbb N}}

\newcommand\si{\sigma}

\newcommand\sN{{\mathcal N}}

\renewcommand\ell{l}

\newcommand\bW{\mathbf{W}}
\newcommand\bA{\mathbf{A}}
\newcommand\bK{\mathbf{K}}
\newcommand\pr{\mathrm{Pr}}

\newcounter{mycount}

\numberwithin{equation}{section}
\numberwithin{theorem}{section}
\numberwithin{figure}{section}

\title[Community Detection in k-Community Gaussian Mixture Model]{Exact Recovery of Community Detection in k-Community Gaussian Mixture Model}

\author{Zhongyang Li}
\address{Department of Mathematics,
University of Connecticut,
Storrs, Connecticut 06269-3009, USA}
\email{zhongyang.li@uconn.edu}
\urladdr{\url{https://mathzhongyangli.wordpress.com}}

\begin{document}
\maketitle

\begin{abstract}
We study the community detection problem on a Gaussian mixture model, in which vertices are divided into $k\geq 2$ distinct communities. The major difference in our model is that the intensities for Gaussian perturbations are different for different entries in the observation matrix, and we do not assume that every community has the same number of vertices.  We explicitly find the threshold for the exact recovery of the maximum likelihood estimation. Applications include the community detection on hypergraphs.
\end{abstract}

\section{Introduction}

Community structures are ubiquitous in graphs modeling natural and social phenomena. In natural sciences, atoms form molecules so that atoms in the same molecule have stronger connections compared to those in different molecules. In social sciences, individuals form groups in such a way that individuals in the same group have more communications compared to individuals in different groups. The main aim for community detection is to determine the specific groups that specific individuals belong to based on observations of (random) connections between individuals. Identifying different communities in the stochastic block model is a central topic in many fields of science and technology; see \cite{EA18} for a summary.

In this paper we study the community detection problem for the Gaussian mixture model, in which there are $n$ vertices belonging to $k$ ($k\geq 2$) different communities. We observe a $p\times 1$ vector for each one of the $n$ vertices, perturbed by a $p\times 1$ Gaussian vector with independent (but not necessarily identically distributed), mean-0 entries. More precisely, each entry of the $p\times n$ perturbation matrix is obtained by a multiple of a standard Gaussian random variable, while the intensities of different entries are different. Given such an observation, we find the maximum likelihood estimation (MLE) for the community assignment, and study the probability that the MLE equals the true community assignment as the number of vertices $n\rightarrow\infty$. If this probability tends to $1$ as $n\rightarrow\infty$, we say exact recovery occurs. 
Heuristically, it is natural to conjecture that exact recovery may occur when the intensities of the perturbations are small but does not occur when these intensities are large. The major theme of the paper is to investigate how small the intensities of the perturbations are needed in order to ensure the exact recovery, and how large the intensities are required to stop the occurrences of the exact recovery.

Clustering problems in the Gaussian mixture model has been studied extensively, see \cite{JB67,DLR77,DP81,FR02} for an incomplete list. We mention some recent related work here.

The Gaussian mixture model when all the entries of the perturbation matrix are i.i.d was studied in \cite{CY20}, in which a condition for the exact recovery of the semi-definite programming is proved. When all the communities have the same size, a condition that exact recovery does not occur was also proved in \cite{CY20} when the number of communities $k\leq \log n$. The case of unbalanced communities was investigated in \cite{GV18}. In this paper, we obtain conditions when the exact recovery happens and does not happen for the more general Gaussian mixture model when the entries of the perturbation matrix are not necessarily identically distributed. Our result can be applied to the special case when intensities of the Gaussian perturbations are all equal, and in particular, we obtain a condition that the exact recovery of MLE does not occur when the number of communities $k$ is $e^{o(\log n)}$ in the hypergraph model, for any fixed constant $C_2$ independent of $n$.

When $p=n$ in our model, we may consider the rows and columns of the observation matrix are indexed by the $n$ vertices, and each entry represents an edge. In this case we obtain the community detection problem on a graph.  When $p=n^s$ with $s\geq 2$, we may consider the rows of the observation matrix are indexed by ordered $s$-tuples of vertices, and each entry of the observation matrix represents a $(s+1)$-hyperedge. In this case we obtain the community detection problem on a hypergraph. Community detections on hypergraphs with Gaussian perturbations were studied in \cite{KBG}, where the vertices are divided into two equal-sized communities, and a weight-1 $(d+1)$-hyperedge exists if and only if all the vertices are in the same group.
The results proved in this paper can be applied to the community detection problems on hypergraphs with Gaussian perturbation to obtain necessary and sufficient conditions for the exact recovery, in which the number of communities is arbitrary and communities are not necessarily equal-sized; moreover the hyperedges have general weights as represented in the (unperturbed) observation matrix. Community detection problems on random graphs were also studied in \cite{HLL83,DF89,MNS13,LM14,ABH15,AS18}.

The organization of the paper is as follows. In section \ref{bmr}, we review the definition of the Gaussian mixture models and hypergraphs, and state the main results proved in this paper. In section \ref{gmec}, we prove conditions for the exact recovery of the Gaussian mixture model when the number of vertices in each community is unknown. In Section \ref{hg1}, we apply the results proved in section \ref{gmec} to the exact recovery of the community detection in hypergraphs, and also prove conditions when exact recovery does not occur in hypergraphs under the assumption that the number of vertices in each community is unknown. In section \ref{gm2}, we prove conditions for the exact recovery of the Gaussian mixture model when the number of vertices in each community is known and fixed. In Section \ref{hg2}, we prove conditions when exact recovery does not occur in hypergraphs under the assumption that the number of vertices in each community is known and fixed. In Section \ref{adl}, we prove a lemma used to obtain the main results of the paper.

\section{Backgrounds and Main Results}\label{bmr}

In this section, we review the definition of the Gaussian mixture models and hypergraphs, and state the main results proved in this paper. 

\subsection{Gaussian mixture model}

Let $n\geq k\geq 2$ be positive integers. Let 
\begin{eqnarray*}
[n]=\{1,2,\ldots,n\}
\end{eqnarray*}
be a set of $n$ vertices divided into $k$ different communities. Let
\begin{eqnarray*}
[k]:=\{1,\ldots,k\}
\end{eqnarray*}
be the set of communities. A mapping $x: [n]\rightarrow [k]$ which assigns a unique community represented by an integer in $[k]$ to each one of the $n$ vertices in $[n]$ is called a community assignment mapping.  Let $\Omega$ be the set consisting of all the possible  mappings from $[n]$ to $[k]$; i.e.
\begin{eqnarray*}
\Omega:=\{x:[n]\rightarrow[k]\}.
\end{eqnarray*}
Each mapping in $\Omega$ is a community assignment mapping.

Let $p\geq 1$ be a positive integer. Let 
\begin{eqnarray*}
\theta:\Omega\times[p]\times [k]\rightarrow\RR
\end{eqnarray*}
be a function on the set $\Omega\times [p]\times [k]$ taking real values.

For a community assignment mapping $x\in \Omega$, let $\mathbf{A}_x$ be a $p\times n$ matrix whose entries are given by 
\begin{eqnarray}
(\mathbf{A}_x)_{i,j}=\theta(x,i,x(j)),\ \forall i\in[p], j\in[n].\label{dA}
\end{eqnarray}
Let $\Sigma$ be a $p\times n$ matrix with positive real entries defined by
\begin{eqnarray*}
\Sigma:=(\sigma_{i,j})_{i\in[p],j\in[n]}\in (\RR^{+})^{p\times n}
\end{eqnarray*}

Let $P,Q$ be two $p\times n$ matrices. Define the inner product of $P,Q$ by
\begin{eqnarray*}
\langle P,Q\rangle=\sum_{i\in[p]}\sum_{j\in[n]} P_{i,j}Q_{i,j}.
\end{eqnarray*}
The norm $\|P\|$ for a matrix $P$ is defined by
\begin{eqnarray*}
\|P\|=\sqrt{\langle P, P\rangle}.
\end{eqnarray*}
Let $P*Q$ be a $p\times n$ matrix defined by
\begin{eqnarray*}
P*Q:=(P_{i,j}Q_{i,j})_{i\in[p],j\in[n]}
\end{eqnarray*}

Define a random observation matrix $\bK_x$ by
\begin{eqnarray}
\mathbf{K}_x=\mathbf{A}_x+ \Sigma*\bW;\label{dkx}
\end{eqnarray}
where $\bW$ is a random $p\times n$ matrix with i.i.d.~standard Gaussian entries. Note that if the entries of $\Sigma$ are not all equal, the perturbation matrix $\Sigma*W$ has independent but not identically distributed entries.

Let $y\in \Omega$ be the true community assignment mapping. Given the observation $\bK_y$, the goal is to determine the true community assignment mapping $y$. We shall apply the maximum likelihood estimation (MLE).  

Let $n_1,\ldots,n_k$ be positive integers satisfying
\begin{eqnarray*}
\sum_{i=1}^{k}n_i=n. 
\end{eqnarray*}
and
\begin{eqnarray*}
|y^{-1}(i)|=n_i,\ \forall i\in [k];
\end{eqnarray*}
i.e. $n_i$ is the number of vertices in community $i$ for each $i\in[k]$ under the mapping $y$.

Let
\begin{eqnarray*}
\Omega_{n_1,\ldots,n_k}:=\{x\in \Omega: |x^{-1}(i)|=n_i,\ \forall i\in[k]\}
\end{eqnarray*}
be the set of all the community assignment mappings such that there are exactly $n_i$ vertices in the community $i$, for each $i\in[k]$.

For each real number $c\in(0,1)$, let
\begin{eqnarray*}
\Omega_c:=\left\{x\in \Omega: \frac{|x^{-1}(i)|}{\sum_{j\in[k]}|x^{-1}(j)|}\geq c,\ \forall i\in[k]\right\},
\end{eqnarray*}
i.e. $\Omega_c$ consists of all the community assignment mappings such that the percentage of the numbers of vertices in each community is at least $c$.

Assume the true community assignment mapping $y\in \Omega_c$ for some $c\in(0,1)$. Let $\Phi$ be an $p\times n$ matrix whose entries are given by
\begin{eqnarray*}
(\Phi)_{i,j}=\frac{1}{\sigma_{i,j}},\ \forall i\in[p], j\in[n];
\end{eqnarray*}
in other words, the $(i,j)$-entry of $\Phi$ is the reciprocal of the $(i,j)$-entry of $\Sigma$.
Define
\begin{eqnarray}
\hat{y}:=\mathrm{argmin}_{x\in \Omega_{\frac{2c}{3}}}\|\Phi*(\bK_y-\bA_x)\|^2\label{dhy}
\end{eqnarray}
and
\begin{eqnarray}
\check{y}:=\mathrm{argmin}_{x\in \Omega_{n_1,\ldots,n_k}}\|\Phi*(\bK_y-\bA_x)\|^2\label{dcy}
\end{eqnarray}

Then we have the following lemma 
\begin{lemma}
$\hat{y}$ is the MLE with respect to the observation $\bK_y$ in $\Omega_{\frac{2c}{3}}$. $\check{y}$ is the MLE with respect to the observation $\bK_y$ in $\Omega_{n_1,\ldots,n_k}$. 
\end{lemma}

\begin{proof}By definition, the MLE with respect to the observation $\bK_y$ in $\Omega_{\frac{2c}{3}}$ (resp.\ $\Omega_{n_1,\ldots,n_k}$) should maximize the probability density of the observation $\bK_y$ among all $x\in \Omega_{\frac{2c}{3}}$ (resp.\ $x\in \Omega_{n_1,\ldots,n_k}$). If the true community assignment mapping $y=x$, we may consider $\bK_{y}$ as a random matrix with mean value $\bA_x$ and independent entries, such that variance of its $(i,j)$-entry is $\sigma_{i,j}^2$. Therefore the probability density of $\bK_y$ is given by
\begin{eqnarray*}
\left(\prod_{i\in[p],j\in[n]}\frac{1}{\sqrt{2\pi}\si_{i,j}}\right)e^{-\sum_{i\in[p],j\in[n]}\frac{(\bK_y-\bA_x)_{i,j}^2}{2\si_{i,j}^2}},
\end{eqnarray*}
where the exponent is exactly
\begin{eqnarray*}
-\frac{1}{2}\|\Phi*(\bK_y-\bA_x)\|^2.
\end{eqnarray*}
It is straightforward to check that the minimizer of $\|\Phi*(\bK_y-\bA_x)\|^2$ is exactly the maximizer of the probability density. Then the lemma follows.
\end{proof}

We shall investigate under which conditions we have $\check{y}=y$ and $\hat{y}=y$ with high probability.

To state the main theorems proved in this paper, we first introduce a few assumptions.

For $x,y\in \Omega$, let
\begin{eqnarray}
L_{\Phi}(x,y):=\|\Phi*(\bA_x-\bA_y)\|^2.\label{lxy}
\end{eqnarray}

For $x\in \Omega$, let
\begin{eqnarray*}
n_i(x)=|x^{-1}(i)|,\ \forall\ i\in [k];
\end{eqnarray*}
then $n_i(x)$ is the number of vertices in community $i$ under the community assignment mapping $x$.
It is straightforward to check that
\begin{eqnarray*}
\sum_{i=1}^k n_i(x)=n.
\end{eqnarray*}

For $i,j\in[k]$ and $x,z\in\Omega$, 
let $t_{i,j}(x,z)$ be a nonnegative integer given by
\begin{eqnarray*}
t_{i,j}(x,z)=|x^{-1}(i)\cap z^{-1}(j)|.
\end{eqnarray*}
That is, $t_{i,j}(x,z)$ is the number of vertices in $[n]$ which are in community $i$ under the mapping $x$ and in community $j$ under the mapping $z$.
Then
\begin{eqnarray}
\sum_{j\in[k]}t_{i,j}(x,z)=n_i(x);\qquad \sum_{i\in[k]}t_{i,j}(x,z)=n_j(z);\label{tn}
\end{eqnarray}

Define a set 
 \begin{eqnarray}
 \mathcal{B}:=\left\{(t_{1,1},t_{1,2},\ldots,t_{k,k})\in\{0,1,2,\ldots,n\}^{k^2}:\sum_{i=1}^{k}t_{i,j}=n_j\right\}.\label{dsb}
 \end{eqnarray}
 For $\epsilon>0$, define a set $\sB_{\epsilon}$ consisting of all the $(t_{1,1},t_{1,2},\ldots,t_{k,k})\in \sB$ satisfying all the following conditions:
 \begin{enumerate}
\item $\forall\ i\in[k],\  \max_{j\in[k]}t_{j,i}\geq n_i-n\epsilon$.
\item For $i\in[k]$, let $t_{w(i),i}=\max_{j\in[k]}t_{j,i}$. Then $w$ is a bijection from $[k]$ to $[k]$.
\item $w$ is $\theta$-preserving, i.e. for any $x\in\Omega$, $i\in[p]$ and $a\in[k]$, we have
\begin{eqnarray*}
 \theta(x,i,a)=\theta(w\circ x,i,w(a)).
 \end{eqnarray*}
\end{enumerate}

We may assume $\theta$ and $\Sigma$ satisfy the following assumptions.
 
 \begin{assumption}\label{ap27}
 \begin{enumerate}
 \item There exists $B_1>0$, such that for all $i,j\in[p]\times n$, we have
 \begin{eqnarray*}
 |\si_{i,j}|\leq B_1.
 \end{eqnarray*}
 \item Assume $\epsilon\in(0,\frac{2c}{3k})$, $x\in \Omega_{\frac{2c}{3}}$ and $y\in \Omega_c$. Then for all $x,y\in \Omega$, and
 \begin{eqnarray}
 (t_{1,1}(x,y),t_{1,2}(x,y),\ldots,t_{k,k}(x,y))\in \mathcal{B}\setminus \mathcal{B}_{\epsilon},\label{tcd}
 \end{eqnarray}
 we have
 \begin{eqnarray*}
 \sum_{i\in[p],j\in[n]}(\theta(x,i,x(j))-\theta(y,i,y(j)))^2\geq T(n)
 \end{eqnarray*}
 \end{enumerate}
 \end{assumption}
 
 We now introduce an equivalence condition on $\Omega$.

\begin{definition}\label{dfeq}
For $x\in \Omega$, let $C(x)$ consist of all the $x'\in \Omega$ such that $x'$ can be obtained from $x$ by a $\theta$-preserving bijection of communities.  More precisely, $x'\in C(x)\subset \Omega$ if and only if the following conditions hold
\begin{enumerate}
\item for $i,j\in[n]$, $x(i)=x(j)$ if and only if $x'(i)=x'(j)$; and
\item for $i\in[p]$ and $j\in[n]$, $\theta(x,i,x(j))=\theta(x',i,x'(j))$.
\end{enumerate}
Note that condition (1) above is equivalent of saying that there is a bijection $\eta:[k]\rightarrow [k]$, such that
\begin{eqnarray*}
x=\eta\circ x'
\end{eqnarray*}
where $\circ$ denotes the composition of two mappings.

We define an equivalence relation on $\Omega$ as follows: we say $x,z\in \Omega$ are equivalent if and only if $x\in C(z)$. Let $\ol{\Omega}$ be the set of all the equivalence classes in $\Omega$.
\end{definition}
 
 \begin{assumption}\label{ap214}Assume $\epsilon\in(0,\frac{2c}{3k})$, $x\in \Omega_{\frac{2c}{3}}$ and $y\in \Omega_c$. Assume there exists $\Delta>0$ such that:

Let $y_1,y_2\in\Omega_{\frac{2c}{3}}$ and $a,b\in[k]$ and $a\neq b$. Let $i,j\in[n]$ such that $i\in y_1^{-1}(a)\cap x^{-1}(b)$. Let $y_2:[n]\rightarrow[k]$ be defined as follows
\begin{eqnarray*}
y_2(j):=\begin{cases}b&\mathrm{if}\ j=i\\ y_1(j)&\mathrm{if}\ j\in[n]\setminus\{i\} \end{cases}.
\end{eqnarray*}
When
\begin{eqnarray*}
\left(t_{1,1}(x,y_1),t_{1,2}(x,y_1),\ldots,t_{k,k}(x,y_1)\right)\in \mathcal{B}_{\epsilon}
\end{eqnarray*}
such that for all $i\in[k]$
\begin{eqnarray*}
t_{i,i}=\max_{j\in[k]}t_{j,i}(x,y_1);
\end{eqnarray*}
$\epsilon\in\left(0,\frac{2c}{3k}\right)$; and
\begin{eqnarray*}
y_1\notin C(x);\ 
\end{eqnarray*}
we have
\begin{eqnarray*}
L_{\Phi}(x,y_1)-L_{\Phi}(x,y_2)\geq \Delta(1+o(1)).
\end{eqnarray*}
where $o(1)\rightarrow 0$, as $n\rightarrow\infty$.

\end{assumption}

We may assume that $\theta$ satisfies the following assumption.

\begin{assumption}\label{ap24}Let $x,z\in \Omega$. If for any $i\in[p]$ and $j\in[n]$, 
\begin{eqnarray}
\theta(x,i,x(j))=\theta(z,i,z(j));\label{sxze}
\end{eqnarray}
then $x\in C(z)$.
\end{assumption}

Assumption \ref{ap24} actually says that for two community assignment mappings $x$ and $z$, if they are not equivalent then $\theta\circ x$ and $\theta \circ z$ are different. In other words, it assumes that $\theta$ can distinguish different equivalence classes in $\Omega$. See Section \ref{ep24} for examples.

\begin{theorem}\label{p215}Assume $y\in \Omega_c$ is the true community assignment mapping. Suppose that Assumptions \ref{ap24}, \ref{ap27} and \ref{ap214} hold. Let $\epsilon\in(0,\frac{2c}{3k})$. If
\begin{eqnarray}
\lim_{n\rightarrow\infty} n\log k-\frac{T(n)}{8B_1^2}=-\infty,\label{ld1}
\end{eqnarray}
and for any constant $\delta>0$ independent of $n$,
\begin{eqnarray}
\lim_{n\rightarrow\infty}\log k+\log n-\frac{\Delta(1-\delta)}{8}=-\infty,\label{ld2}
\end{eqnarray}
then $\lim_{n\rightarrow\infty}\Pr(\hat{y}\in C(y))=1$.
\end{theorem}

Theorem \ref{p215} gives a sufficient condition for the exact recovery of MLE in the Gaussian mixture model. It is proved in Section \ref{gmec}. An application of Theorem \ref{p215} on the exact recovery of community detection on hypergraphs is discussed in Section \ref{at42}.

We also obtain a condition for the exact recovery when the sample space of the MLE is restricted to $\Omega_{n_1,\ldots,n_k}$; i.e. the number of vertices in each community is known and fixed.

\begin{assumption}\label{ap46} Assume $x,y_m,y_h\in \Omega$ such that
\begin{enumerate}
\item $D_{\Omega}(y_m,y_h)=j$, where $j\geq 2$ is a positive integer; and
\item There exist $u_1,\ldots,u_j\in[n]$, such that 
\begin{enumerate}
\item $y_m(v)=y_h(v)$, for all $v\in [n]\setminus \{u_1,\ldots,u_j\}$; and 
\item $y_m(u_i)\neq y_h(u_i)=x(u_i)=y_m(u_{i-1})$ for all $i\in[j]$.
\item $(t_{1,1}(x,y_m),t_{1,2}(x,y_m),\ldots,t_{k,k}(x,y_m))\in \mathcal{B}_{\epsilon}$ with $\epsilon\in\left(0,\frac{2c}{3k}\right)$ and $w(i)=i$.
\end{enumerate}
\end{enumerate}
Then 
\begin{eqnarray}
L_{\Phi}(x,y_m)- L_{\Phi}(x,y_h)\geq j\Delta(1+o(1))\label{l46d}
\end{eqnarray}
\end{assumption}

\begin{theorem}\label{m27}Suppose that Assumptions \ref{ap27} \ref{ap46}, (\ref{ld1}) and (\ref{ld2}) hold. Then $\lim_{n\rightarrow\infty}\Pr(\check{y}\in C(y))=1$.
\end{theorem}

Indeed, Assumption \ref{ap214} implies Assumption \ref{ap46}; see Lemma \ref{lm46}. Theorem \ref{m27} is proved in Section (\ref{gm2}).

\subsection{Hypergraphs}

A special case for the Gaussian mixture model is the hypergraph model. Let $s,s_1,s_2$ be positive integers satisfying
\begin{eqnarray*}
2\leq s_1\leq s\leq s_2.
\end{eqnarray*}
 A hypergraph $H=(V,E)$ has vertex set $V:=[n]$ and hyper-edge set $E$ defined as follows: 
 \begin{eqnarray*}
 E:=\{(a_1,\ldots,a_s): a_1,\ldots,a_s\in [n], s\in\{s_1,s_1+1,\ldots,s_2\}\}
 \end{eqnarray*}
 
 Let $\phi:\cup_{s=s_1}^{s_2}[k]^{s}\rightarrow [0,\infty)$ be a function which assigns a unique real number $\phi(c_1,\ldots,c_s)$ to each $s$-tuple of communities $(c_1,\ldots,c_s)\in [k]^s$, and $s\in [s_1,s_2]$. 
 
 For a community assignment mapping $x$, the weighted adjacency tensor $\mathbf{A}_x$ is defined by 
\begin{eqnarray}
(\mathbf{A}_x)_{a_1,\ldots,a_s}=\begin{cases}\phi(x(a_1),\ldots,x(a_s)),&\mathrm{if}\ (a_1,\ldots,a_s)\in E\\ 0&\mathrm{otherwise}.\end{cases}\label{dA}
\end{eqnarray}
and
\begin{eqnarray*}
\Sigma_{(a_1,\ldots,a_s)}:=\sigma_{(a_1,\ldots,a_s)}
\end{eqnarray*}
Define a random tensor $\bK_x$ as in (\ref{dkx}). Recall that $y\in \Omega_c$ is the true community assignment mapping. Define $\hat{y}$ and $\check{y}$ as in (\ref{dhy})
and (\ref{dcy}).

Recall that $y\in \Omega$ is the true community assignment mapping satisfying $|y^{-1}(i)|=n_i$, for all $i\in[k]$. 
 Let $a\in[n]$. 
Let $y^{(a)}\in\Omega$ be defined by
\begin{eqnarray}
y^{(a)}(i)=\begin{cases} y(i)&\mathrm{if}\ i\in[n],\ \mathrm{and}\ i\neq a\\ y^{(a)}(a)&\mathrm{if}\ i=a.\end{cases}\label{dya}
\end{eqnarray}
such that
\begin{eqnarray*}
y(a)\neq y^{(a)}(a)\in[k].
\end{eqnarray*}

\begin{theorem}\label{p31} Assume
\begin{eqnarray}
\lim_{n\rightarrow\infty}\min_{i\in[k]}n_i=\infty.\label{nii}
\end{eqnarray}
Suppose that there exists a subset $H\subset [n]$ satisfying all the following conditions
\begin{enumerate}
\item $|H|=h=o(n)$;
\item $\lim_{n\rightarrow\infty}\frac{\log h}{\log n}=1$;
\item For each $g\in H$,
\begin{eqnarray*}
&&\sum_{s=s_1}^{s_2}\sum_{j=1}^{s}\sum_{(i_1,\ldots,\widehat{i}_j,\ldots,i_s)\in([n]\setminus H)^{s-1}}\frac{1}{\sigma_{(i_1,\ldots,i_{j-1},g,i_{j+1},\ldots,i_s)}^2}\\
&&\times(\phi(y(i_1),\ldots,y^{(g)}(g),\ldots,y(i_s))-\phi(y(i_1),\ldots,y(g),\ldots,y(i_s)))^2\\
&=&(1+o(1))L_{\Phi}(y^{(g)},y)
\end{eqnarray*}
\item there exists a constant $\beta>0$ independent of $n$, such that
\begin{eqnarray*}
\frac{\max_{a\in H}L_{\Phi}(y^{(a)},y)}{\min_{a\in H} L_{\Phi}(y^{(a)},y)}\leq \beta^2,\ \forall n.
\end{eqnarray*}
\end{enumerate}
If there exists a constant $\delta>0$ independent of $n$, such that
\begin{eqnarray}
\max_{a\in H}L_{\Phi}(y^{(a)},y)<8(1-\delta)\log n\label{as31}
\end{eqnarray}
Then $\lim_{n\rightarrow\infty}\Pr(\hat{y}\in C(y))=0$.
\end{theorem}

Theorem \ref{p31} is proved in Section \ref{hg1}. An example is given in Section \ref{at41}.

Let $a,b\in[n]$ such that $y(a)\neq y(b)$. Let $y^{(ab)}\in \Omega_{n_1,\ldots,n_k}$ be the community assignment mapping defined by
\begin{eqnarray}
y^{(ab)}(i)=\begin{cases}y(i)& \mathrm{if}\ i\in[n]\setminus\{a,b\}\\ y(b)&\mathrm{if}\ i=a\\ y(a)& \mathrm{if}\ i=b \end{cases}\label{yab1}
\end{eqnarray}
In other words, $y^{(ab)}$ is obtained from $y$ by exchanging $y(a)$ and $y(b)$.

We also prove a condition when the exact recovery does not occur if the sample space of the MLE is restricted in $\Omega_{n_1,\ldots,n_k}$.

\begin{theorem}\label{p210}Assume
\begin{eqnarray}
\lim_{n\rightarrow\infty}\min_{i\in[k]}n_i=\infty.\label{nii1}
\end{eqnarray}
Suppose that there exist two subsets $H_1,H_2\subset [n]$ satisfying all the following conditions
\begin{enumerate}
\item $|H_1|=|H_2|=h=o(n)$;
\item $\lim_{n\rightarrow\infty}\frac{\log h}{\log n}=1$;
\item For any $u_1,u_2\in H_1$ and $v_1,v_2\in H_2$, 
\begin{eqnarray*}
y(u_1)=y(u_2)\neq y(v_1)=y(v_2);
\end{eqnarray*}
\item For any $u\in H_1$ and $v\in H_2$
\begin{eqnarray*}
&&\sum_{s=s_1}^{s_2}\sum_{j=1}^{s}\sum_{(i_1,\ldots,\widehat{i}_j,\ldots,i_s)\in([n]\setminus (H_1\cup H_2))^{s-1}}\left(\frac{1}{\sigma_{(i_1,\ldots,i_{j-1},u,i_{j+1},\ldots,i_s)}^2}+\frac{1}{\sigma_{(i_1,\ldots,i_{j-1},v,i_{j+1},\ldots,i_s)}^2}\right)\\
&&(\phi(y(i_1),\ldots,y(v),\ldots,y(i_s))-\phi(y(i_1),\ldots,y(u),\ldots,y(i_s)))^2\\
&=&(1+o(1))L_{\Phi}(y^{(uv)},y)
\end{eqnarray*}
\item For any $g\in H_1\cup H_2$, the quantity
\begin{eqnarray*}
&&\sum_{s=s_1}^{s_2}\sum_{j=1}^{s}\sum_{(i_1,\ldots,\widehat{i}_j,\ldots,i_s)\in([n]\setminus (H_1\cup H_2))^{s-1}}\frac{1}{\sigma_{(i_1,\ldots,i_{j-1},g,i_{j+1},\ldots,i_s)}^2}\\
&&(\phi(y(i_1),\ldots,y(b),\ldots,y(i_s))-\phi(y(i_1),\ldots,y(a),\ldots,y(i_s)))^2
\end{eqnarray*}
is a constant and is independent of $g$.
\end{enumerate}
If there exists a constant $\delta>0$ independent of $n$, such that
\begin{eqnarray}
\max_{u\in H_1,v\in H_2}L_{\Phi}(y^{(uv)},y)<16(1-\delta)\log n,\label{as51}
\end{eqnarray} 
$\lim_{n\rightarrow\infty}\Pr(\check{y}\in C(y))=0$.
\end{theorem}

Theorem \ref{p210} is proved in Section \ref{hg2}.

\section{Community Detection on K-Community Gaussian Mixture Models}\label{gmec}

In this section, we consider the MLE when the number of vertices in each community is unknown. We shall obtain a sufficient condition for the occurrence of the exact recovery. The main goal is prove Theorem \ref{p215}.

Recall that we defined an equivalence relation on $\Omega$ in Definition \ref{dfeq}. It is straightforward to check that
\begin{eqnarray*}
\bK_y=\bK_{y'},\ \mathrm{and}\ \bA_y=\bA_{y'},\qquad\mathrm{if}\ y'\in C(y).
\end{eqnarray*}
Therefore, the MLE based on the observation $\bK_y$ can only recover the community assignment mapping up to equivalence.

Note that
\begin{eqnarray}
\langle \Phi*\bA_x, \Phi*\bA_z\rangle&=&\sum_{i\in[p],j\in[n]}\frac{(\bA_x)_{i,j}(\bA_z)_{i,j}}{\sigma_{i,j}^2}\label{lak}\\
&=&\sum_{i\in[p],j\in[n]}\frac{\theta(x,i,x(j))\theta(z,i,z(j))}{\si_{i,j}^2}
\notag
\end{eqnarray}
In particular for each $x\in\Omega$ we have
\begin{eqnarray}
\|\Phi*\bA_x\|^2=\sum_{i\in[p];j\in[n]}\frac{(\theta(x,i,x(j)))^2} {\si_{i,j}^2}\label{ax}
\end{eqnarray}

Recall that $y\in \Omega_{n_1,\ldots,n_k}$  is the true community assignment mapping. Note that
\begin{eqnarray}
\|\Phi*(\bK_y-\bA_x)\|^2&=&\|\Phi*\bK_y\|^2-2\langle \Phi*\bK_y, \Phi*\bA_x\rangle+\|\Phi*\bA_x\|^2\label{kymx}
\end{eqnarray}
For each fixed observation $\bK_y$, $\|\Phi*\bK_y\|^2$ is fixed and independent of $x\in \Omega$. Therefore
\begin{eqnarray*}
\hat{y}:&=&\mathrm{argmin}_{x\in \Omega_{\frac{2c}{3}}}\|\Phi*(\bK_y-\bA_x)\|^2\\
&=&\mathrm{argmin}_{x\in \Omega_{\frac{2c}{3}}}\left(-2\langle \Phi*\bK_y, \Phi*\bA_x\rangle+\|\Phi*\bA_x\|^2\right)
\end{eqnarray*}
For $x\in \Omega$, define
\begin{eqnarray}
f(x):=-2\langle \Phi*\bK_y, \Phi*\bA_x\rangle+\|\Phi*\bA_x\|^2\label{df}
\end{eqnarray}
Then 
\begin{eqnarray}
f(x)-f(y)\label{fxmy}&=&\|\Phi*\bA_x\|^2-\|\Phi*\bA_y\|^2\\
&&-2\langle \Phi*\bA_y, \Phi*(\bA_x-\bA_y)\rangle
-2\langle \bW, \Phi*(\bA_x-\bA_y) \rangle\notag\\
&=&\|\Phi*(\bA_x-\bA_y)\|^2-2\langle \bW, \Phi*(\bA_x-\bA_y)\notag,
\end{eqnarray}
where we use the identity
\begin{eqnarray*}
\Phi*(\Sigma*\bW)=\bW.
\end{eqnarray*}
Then $f(x)-f(y)$ is a Gaussian random variable with mean value
\begin{eqnarray*}
\mathbb{E}\left(f(x)-f(y)\right)=L_{\Phi}(x,y);
\end{eqnarray*}
and variance
\begin{eqnarray*}
\mathrm{Var}(f(x)-f(y))
&=&4L_{\Phi}(x,y).
\end{eqnarray*}

\begin{lemma}\label{lfe}For $x,z\in \Omega$. If $x\in C(z)$, then
\begin{eqnarray*}
f(x)=f(z).
\end{eqnarray*}
\end{lemma}
\begin{proof}By Definition \ref{dfeq}, if $x\in C(z)$, then for any $i\in[p]$ and $j,h\in[n]$, $x(j)=x(h)$ if and only if $z(j)=z(h)$ and $\theta(x,i,x(j))=\theta(z,i,z(j))$, then $\bA_x=\bA_z$ by (\ref{dA}). Moreover, since for any $i\in [p]$ and $j\in[n]$, 
\begin{eqnarray*}
(\bA_x)_{i,j}=(\bA_z)_{i,j};
\end{eqnarray*}
we have
\begin{eqnarray*}
\frac{(\bA_x)_{i,j}}{\si_{i,j}}=\frac{(\bA_z)_{i,j}}{\si_{i,j}};
\end{eqnarray*}
this implies 
\begin{eqnarray}
\Phi*\bA_x=\Phi*\bA_z. \label{paez}
\end{eqnarray}
Then the lemma follows from (\ref{df}).
\end{proof}

Define
\begin{eqnarray*}
p(\hat{y};\sigma):&=&\mathrm{Pr}\left(\hat{y}\in C(y)\right)=\pr\left(f(y)<\min_{C(x)\in \ol{\Omega}_{\frac{2c}{3}},C(x)\neq C(y)}f(x)\right)
\end{eqnarray*}
Then
\begin{eqnarray*}
1-p(\hat{y};\sigma)&\leq& \sum_{C(x)\in\ol{\Omega}_{\frac{2c}{3}}:C(x)\neq C(y)}\pr(f(x)-f(y)\leq 0)\\
&=& \sum_{C(x)\in\ol{\Omega}_{\frac{2c}{3}}:C(x)\neq C(y)}\pr_{\xi\in\mathcal{N}(0,1)}\left(\xi\leq \frac{-\sqrt{L_{\Phi}(x,y)}}{2 }\right)\\
&\leq& \sum_{C(x)\in\ol{\Omega}_{\frac{2c}{3}}:C(x)\neq C(y)} e^{-\frac{(L_{\Phi}(x,y))^2}{8}}.
\end{eqnarray*}

\begin{lemma}\label{l12}Let $x,y,x',y'\in \Omega$, such that $x'\in C(x)$ and $y'\in C(y)$, then 
\begin{eqnarray*}
L_{\Phi}(x,y)=L_{\Phi}(x',y').
\end{eqnarray*}
\end{lemma}

\begin{proof}By (\ref{paez}) we obtain that when $x'\in C(x)$ and $y'\in C(y)$,
\begin{eqnarray*}
\Phi*\bA_x=\Phi*\bA_{x'};\qquad \Phi*\bA_y=\Phi*\bA_{y'}.
\end{eqnarray*}
Then the lemma follows from (\ref{lxy}).
\end{proof}

\begin{lemma}\label{lgz}For $x,y\in\Omega$, $L_{\Phi}(x,y)\geq 0$. Moreover
\begin{enumerate}
\item If $x\in C(y)$, then $L_{\Phi}(x,y)=0$.
\item If $\theta$ satisfies Assumption \ref{ap24} and $L_{\Phi}(x,y)=0$, then $x\in C(y)$.
\end{enumerate}
\end{lemma}

\begin{proof}From (\ref{lxy}),  it is straightforward to check that $L_{\Phi}(x,y)\geq 0$ for any $x,y\in\Omega$.
Moreover, from (\ref{lxy}) we obtain
\begin{eqnarray*}
L_{\Phi}(x,y)=\sum_{i\in [p],j\in[n]}\si_{i,j}^2(\theta(x,i,x(j))-\theta(y,i,y(j)))^2
\end{eqnarray*}
By the fact that $\si_{i,j}>0$ for all $i\in[p],j\in[n]$, we obtain that $L_{\Phi}(x,y)=0$ if and only if
\begin{eqnarray}
\theta(x,i,x(j))=\theta(y,i,y(j)),\ \qquad\forall i\in[p], j\in[n].\label{eqz}
\end{eqnarray}
If $x\in C(y)$, then there exists a $\theta$-preserving bijection $\eta:[k]\rightarrow[k]$, such that $x=\eta\circ y$. Then (\ref{eqz}) holds by the $\theta$-preserving property of $\eta$, then we obtain Part(1).

On the other hand, if $L_{\Phi}(x,y)=0$, we have (\ref{eqz}) holds. Then $x\in C(y)$ follows from Assumption \ref{ap24}.
 \end{proof}

  \begin{lemma}\label{ll56}Assume that $y\in\Omega_c$ and $x\in \Omega_{\frac{2c}{3}}$.
For $i\in [k]$,  let
\begin{eqnarray}
t_{w(i),i}(x,y)=\max_{j\in [k]}t_{j,i}(x,y),\label{twi}
\end{eqnarray}
where $w(i)\in [k]$. When $\epsilon\in\left(0,\frac{2c}{3k}\right)$ and $(t_{1,1}(x,y),t_{1,2}(x,y),\ldots, t_{k,k}(x,y))\in \RR^{k^2}$ satisfies
\begin{eqnarray*}
\max_{j\in[k]}t_{j,i}(x,y)\geq n_i-n\epsilon,\ \forall i\in [k]
\end{eqnarray*}
$w$ is a bijection from $[k]$ to $[k]$.
\end{lemma}
\begin{proof}See Lemma 5.6 of \cite{ZL19}.
\end{proof}

\begin{definition}\label{df16}Define the distance function $D_{\Omega}:\Omega\times\Omega\rightarrow [n]$ as follows
\begin{eqnarray*}
D_{\Omega}(x,y)=\sum_{i,j\in[k],i\neq j}t_{i,j}(x,y).
\end{eqnarray*}
for $x,y\in\Omega$.
\end{definition}

From Definition \ref{df16}, it is straightforward to check that
\begin{eqnarray*}
D_{\Omega}(x,y)=n-\sum_{i\in[k]}t_{i,i}(x,y)
\end{eqnarray*}

 \begin{lemma}\label{l13} Assume that $\theta,\Sigma$ satisfies Assumptions \ref{ap27}.  Then for all the $x,y\in\Omega$ such that (\ref{tcd}) holds, we have
 \begin{eqnarray*}
 L_{\Phi}(x,y)\geq \frac{T(n)}{B_1^2}.
 \end{eqnarray*}
 \end{lemma}
 
 \begin{proof}Note that
 \begin{eqnarray*}
 L_{\Phi}(x,y)=\left(\sum_{i\in[p],j\in[n]}\left(\theta(x,i,x(j))-\theta(y,i,y(j))\right)^2\right)\left(\frac{\sum_{i\in[p],j\in[n]}\frac{1}{\si_{i,j}^2}\left(\theta(x,i,x(j))-\theta(y,i,y(j))\right)^2}{\sum_{i\in[p],j\in[n]}\left(\theta(x,i,x(j))-\theta(y,i,y(j))\right)^2}\right)
 \end{eqnarray*}
 By Assumption \ref{ap27}(1), we have
 \begin{eqnarray*}
\frac{\sum_{i\in[p],j\in[n]}\frac{1}{\si_{i,j}^2} \left(\theta(x,i,x(j))-\theta(y,i,y(j))\right)^2}{\sum_{i\in[p],j\in[n]}\left(\theta(x,i,x(j))-\theta(y,i,y(j))\right)^2}\geq \frac{1}{B_1^2}
 \end{eqnarray*}
 Then the lemma follows from Assumption \ref{ap27}(2).
 \end{proof}

\noindent{\textbf{Proof of Theorem \ref{p215}}.}
 Note that
\begin{eqnarray*}
\sum_{C(x)\in\ol{\Omega}\setminus C(y)}e^{-\frac{L_{\Phi}(x,y)}{8}}\leq I_1+I_2
\end{eqnarray*}

where 
\begin{eqnarray*}
I_1&=&\sum_{C(x)\in\ol{\Omega}_{\frac{2c}{3}}:\left(t_{1,1}(x,y),\ldots,t_{k,k}(x,y)\right)\in[\mathcal{B}\setminus \mathcal{B}_{\epsilon}], C(x)\neq C(y)}e^{\frac{-(L_{\Phi}(x,y))^2}{8}}
\end{eqnarray*}
and
\begin{eqnarray*}
I_2=\sum_{C(x)\in\ol{\Omega}_{\frac{2c}{3}}:\left(t_{1,1}(x,y),\ldots, t_{k,k}(x,y)\right)\in \mathcal{B}_{\epsilon}, C(x)\neq C(y)}e^{\frac{-(L_{\Phi}(x,y))^2}{8}}.
\end{eqnarray*}
and $\epsilon\in\left(0,\frac{2c}{3k}\right)$.

By Lemma \ref{l13}, when Assumption \ref{ap27} holds, we have
\begin{eqnarray*}
I_1&\leq& k^n e^{-\frac{T(n)}{8B_1^2}}
\end{eqnarray*}
When (\ref{ld1}) holds, we obtain
\begin{eqnarray}
\lim_{n\rightarrow\infty}I_1=0.\label{pr1}
\end{eqnarray}

Now let us consider $I_2$. Let $w$ be the bijection from $[k]$ to $[k]$ as defined in (\ref{twi}). Let $y^*\in \Omega$ be defined by
\begin{eqnarray*}
y^*(z)=w(y(z)),\ \forall z\in [n].
\end{eqnarray*}
Then $y^*\in C(y)$ since $w$ is $\theta$-preserving by the definition of $\mathcal{B}_{\epsilon}$. Moreover, $x$ and $y^*$ satisfies
\begin{eqnarray}
t_{i,i}(x,y^*)\geq n_i(y^*)-n\epsilon,\ \forall i\in[k].\label{dxys}
\end{eqnarray}

We consider the following community changing process to obtain $x$ from $y^*$.
\begin{enumerate}
\item If for all $(j,i)\in [k]^2$, and $j\neq i$, $t_{j,i}(x,y^*)=0$, then $x=y^*$.
\item If (1) does not hold, find the least $(j,i)\in[k]^2$ in lexicographic order such that $j\neq i$ and $t_{j,i}(x,y^*)>0$. Choose an arbitrary vertex $u\in\left\{ x^{-1}(j)\cap (y^*)^{-1}(i)\right\}$. Define $y_1\in \Omega$ as follows
\begin{eqnarray*}
y_1(z)= \begin{cases}j&\mathrm{if}\ z=u\\ y^*(z)&\mathrm{if}\ z\in [n]\setminus \{u\} \end{cases}
\end{eqnarray*}
\end{enumerate}

Then we have
\begin{eqnarray}
&&t_{j,i}(x,y_1)=t_{j,i}(x,y^*)-1\label{tjid}\\
&&t_{j,j}(x,y_1)=t_{j,j}(x,y^*)+1\label{tjjd}\\
&&t_{a,b}(x,y_1)=t_{a,b}(x,y^*)\ \forall (a,b)\in \left([k]^2\setminus\{(j,i),(j,j)\}\right).\notag\\
&&n_i(y_1)=n_i(y^*)-1\notag\\
&&n_j(y_1)=n_j(y^*)+1\notag\\
&&n_{l}(y_1)=n_l(y^*)\ \forall l\in [k]\setminus\{i,j\}.\notag
\end{eqnarray}
Therefore $x$, $y_1$ and $y^*$ satisfy
\begin{eqnarray*}
&&t_{l,l}(x,y_1)\geq n_l(y_1)-n\epsilon;\\
&&t_{l,l}(x,y_1)\geq n_l(y^*)-n\epsilon;\\
&&n_l(y_1)\geq n_l(y^*)-n\epsilon;
\end{eqnarray*}
for all $l\in[k]$.

From Assumption \ref{ap214} and Lemma \ref{l12} we obtain
\begin{eqnarray*}
L_{\Phi}(x,y_1)-L_{\Phi}(x,y)&=&L_{\Phi}(x,y_1)-L_{\Phi}(x,y^*)\leq -\Delta(1+o(1)).
\end{eqnarray*}
Therefore 
\begin{eqnarray}
e^{-\frac{L_{\Phi}(x,y)}{8}}\leq e^{-\frac{L_{\Phi}(x,y_1)}{8}}e^{-\frac{\Delta(1+o(1))}{8}}\label{idi}
\end{eqnarray}

In general, if we have constructed $y_l\in \Omega$ ($r\geq 1$) satisfying all the following conditions:
\begin{eqnarray}
&&t_{l,l}(x,y_r)\geq n_l(y_r)-n\epsilon;\notag\\
&&t_{l,l}(x,y_r)\geq n_l(y^*)-n\epsilon;\notag\\
&&n_l(y_r)\geq n_l(y^*)-n\epsilon;\label{rs}
\end{eqnarray}
for all $l\in[k]$.
We now construct $y_{r+1}\in \Omega$ as follows.
\begin{enumerate}[label=(\alph*)]
\item If for all $(j,i)\in [k]^2$, and $j\neq i$, $t_{j,i}(x,y_r)=0$, then $x=y_r$; then the construction process stops at this step.
\item If (a) does not hold, find the least $(j,i)\in[k]^2$ in lexicographic order such that $j\neq i$ and $t_{j,i}(x,y_r)>0$. Choose an arbitrary vertex $u\in \left\{x^{-1}(j)\cap y_r^{-1}(i)\right\}$. Define $y_{r+1}\in \Omega$ as follows
\begin{eqnarray*}
y_{r+1}(z)= \begin{cases}j&\mathrm{if}\ z=u\\ y_r(z)&\mathrm{if}\ z\in [n]\setminus \{u\} \end{cases}
\end{eqnarray*}
\end{enumerate}
Then it is straightforward to check that
\begin{eqnarray*}
&&t_{l,l}(x,y_{r+1})\geq n_l(y_{r+1})-n\epsilon;\\
&&t_{l,l}(x,y_{r+1})\geq n_l(y^*)-n\epsilon;\\
&&n_l(y_{r+1})\geq n_l(y^*)-n\epsilon;
\end{eqnarray*}
for all $l\in[k]$.

Then if (\ref{dxys}) holds with $y^*$ replaced by $y_r$, then (\ref{dxys}) holds with $y^*$ replaced by $y_{r+1}$. By Assumption \ref{ap214} we obtain
\begin{eqnarray*}
e^{-\frac{L_{\Phi}(x,y_r)}{8}}&\leq& e^{-\frac{L_{\Phi}(x,y_{r+1})}{8}} e^{-\frac{\Delta(1+o(1))}{8}}.
\end{eqnarray*}

Recall that the distance $D_{\Omega}$ in $\Omega$ is defined in Definition \ref{df16}. From the constructions of $y_{r+1}$ we have
\begin{eqnarray*}
D_{\Omega}(x,y_{r+1})=D_{\Omega}(x,y_r)-1.
\end{eqnarray*}
Therefore there exists $h\in [n]$, such that $y_h=x$. By (\ref{idi}) and Assumption \ref{ap214} we obtain
\begin{eqnarray*}
e^{-\frac{L_{\Phi}(x,y)}{8}}\leq e^{-\frac{h\Delta(1+o(1))}{8}}.
\end{eqnarray*}

Since any $x$ in $\mathcal{B}_{\epsilon}$ can be obtained from $y$ by the community changing process described above, we have
\begin{eqnarray}
I_2\leq \sum_{l=1}^{\infty} (nk)^{l} e^{-\frac{l\Delta(1+o(1))}{8}};\label{ifb}
\end{eqnarray}
The right hand side of (\ref{ifb}) is the sum of geometric series with both initial term and common ratio equal to 
\begin{eqnarray}
V:=e^{\log k+\log n-\frac{\Delta(1+o(1))}{8}}\label{ves}
\end{eqnarray}

When (\ref{ld2}) holds, we obtain
\begin{eqnarray}
\lim_{n\rightarrow\infty}I_2=0\label{pr2}
\end{eqnarray}
Then the proposition follows from (\ref{pr1}) and (\ref{pr2}).
$\hfill\Box$

\section{Community Detection on k-Community Hypergraphs}\label{hg1}

In this section, we apply the results proved in section \ref{gmec} to the exact recovery of the community detection in hypergraphs, and also prove conditions when exact recovery does not occur in hypergraphs under the assumption that the number of vertices in each community is unknown.

In the case of a hypergraph, from (\ref{lak}), when 
\begin{eqnarray*}
&&i=(i_1,i_2,\ldots,i_{s-1})\in[n]^{s-1};\\
&&\theta(x,i,a)=\phi(x(i_1),\ldots,x(i_{s-1}),a);
\end{eqnarray*}
we obtain for $x,z\in \Omega$
\begin{eqnarray}
&&\langle \Phi*\bA_x, \Phi*\bA_z\rangle\label{abxz}\\
&=&\sum_{s=s_1}^{s_2}\sum_{(i_1,\ldots,i_s)\in[n]^s} \frac{(\bA_x)_{(i_1,\ldots,i_s)}(\bA_z)_{(i_1,\ldots,i_s)}}{\si_{(i_1,\ldots,i_s)}^2}\notag\\
&=&\sum_{s=s_1}^{s_2}\sum_{(i_1,\ldots,i_s)\in[n]^s} \frac{\phi(x(i_1),\ldots,x(i_s))\phi(z(i_1),\ldots,z(i_s))}{\si^2_{(i_1,\ldots,i_s)}}\notag
\end{eqnarray}
In particular,
\begin{eqnarray*}
\|\Phi*\bA_x\|^2&=&\langle\Phi* \bA_x,\Phi*\bA_x \rangle\\
&=&\sum_{s=s_1}^{s_2}\sum_{(i_1,\ldots,i_s)\in[n]^s} \frac{(\phi(x(i_1),\ldots,x(i_s)))^2}{\si^2_{(i_1,\ldots,i_s)}}
\end{eqnarray*}
Recall that $y\in \Omega_c$ is the true community assignment mapping. Then
\begin{eqnarray*}
\hat{y}&=&\mathrm{argmin}_{x\in \Omega_{\frac{2c}{3}}}\|\Phi*(\mathbf{K}_y-\bA_x)\|^2=\mathrm{argmin}_{x\in \Omega_{\frac{2c}{3}}}f(x)
\end{eqnarray*}
where $f(x)$ is given by (\ref{df}).

By (\ref{fxmy}), we obtain that in the hypergraph case
\begin{eqnarray}
&&f(x)-f(y)\label{fhxy}\\
&=&\|\Phi*(\bA_x-\bA_y)\|^2-2\langle \bW, \Phi*(\bA_x-\bA_y)\rangle\notag\\
&=&\sum_{s=s_1}^{s_2}\sum_{(i_1,\ldots,i_s)\in[n]^s}\frac{(\phi(x(i_1),\ldots,x(i_s))-\phi(y(i_1),\ldots,y(i_s)))^2}{\si_{(i_1,\ldots,i_s)}^2}\notag\\
&&-2\langle \bW, \Phi*(\bA_x-\bA_y)\rangle\notag
\end{eqnarray}
Then $f(x)-f(y)$ is a Gaussian random variable with mean value $L_{\Phi}(x,y)$ and variance $4L_{\Phi}(x,y)$, where $L_{\Phi}(x,y)$ is defined by (\ref{lxy}).

\bigskip

\noindent{\textbf{Proof of Theorem \ref{p31}.} When $y^{(a)}\in \Omega$ is defined by (\ref{dya}),
\begin{eqnarray}
&&t_{y^{(a)}(a),y(a)}(y^{(a)},y)=1;\label{t1}\\
&&t_{y(a),y(a)}(y^{(a)},y)=n_{y(a)}-1;\label{t2}\\
&&t_{i,i}(y^{(a)},y)=n_i;\ \forall\ i\in [k]\setminus \{y(a)\};\label{t3}\\
&&t_{i,j}(y^{(a)},y)=0;\ \forall (i,j)\in [k]^2\setminus\{(y^{(a)}(a),y(a))\},\ \mathrm{and}\ i\neq j.\label{t4}
\end{eqnarray}
and
\begin{eqnarray*}
&&n_{y^{(a)}(a)}(y^{(a)})=n_{y^{(a)}(a)}+1;\\
&&n_{y(a)}(y^{(a)})=n_{y(a)}-1;\\
&&n_i(y^{(a)})=n_i;\ \forall\ i\in[k]\setminus \{y^{(a)}(a),y(a)\}.
\end{eqnarray*}
Moreover,
\begin{eqnarray*}
1-p(\hat{y};\sigma)\geq \mathrm{Pr}\left(\cup_{a\in[n]}\{f(y^{(a)})-f(y)<0\}\right)
\end{eqnarray*}
Since any of the event $\{f(y^{(a)})-f(y)<0\}$ implies $\hat{y}\neq y$.

Let $H\subset [n]$ be given as in the assumptions of the proposition.
Under Assumption (3) of the proposition when $a\in H$ we have 
\begin{eqnarray*}
&&\|\Phi*(\bA_{y^{(a)}}-\bA_y)\|^2\\
&=&\sum_{s=s_1}^{s_2}\sum_{(i_1,\ldots,i_s)\in[n]^s}\frac{\left(\phi(y^{(a)}(i_1)),\ldots,y^{(a)}(i_s))-\phi(y(i_1),\ldots,y(i_s))\right)^2}{\si^2_{(i_1,\ldots,i_s)}}\\
&=&L_{\Phi}(y^{(a)},y)\\
&=&(1+o(1))\left\{\sum_{s=s_1}^{s_2}\sum_{j=1}^{s}\sum_{(i_1,\ldots,\widehat{i}_j,\ldots,i_s)\in([n]\setminus H)^{s-1}}\frac{1}{\sigma_{(i_1,\ldots,i_{j-1},a,i_{j+1},\ldots,i_s)}^2}\right.\\
&&\left.\times(\phi(y(i_1),\ldots,y^{(a)}(a),\ldots,y(i_s))-\phi(y(i_1),\ldots,y(a),\ldots,y(i_s)))^2\right\}
\end{eqnarray*}

Then from (\ref{fhxy}) we have
\begin{eqnarray*}
&&f(y^{(a)})-f(y)\\
&=&-2\langle\mathbf{W},\Phi*\mathbf{A}_{y^{(a)}}-\mathbf{A}_{y} \rangle+(1+o(1))\left\{\sum_{s=s_1}^{s_2}\sum_{j=1}^{s}\sum_{(i_1,\ldots,\widehat{i}_j,\ldots,i_s)\in([n]\setminus H)^{s-1}}\frac{1}{\sigma_{(i_1,\ldots,i_{j-1},a,i_{j+1},\ldots,i_s)}^2}\right.\\
&&\left.\times(\phi(y(i_1),\ldots,y^{(a)}(a),\ldots,y(i_s))-\phi(y(i_1),\ldots,y(a),\ldots,y(i_s)))^2\right\}.
\end{eqnarray*}

 Then $1-p(\hat{y};\sigma)$ is at least
\begin{eqnarray*}
&&\mathrm{Pr}\left(\cup_{a\in[n]}\left\{f(y^{(a)})-f(y)<0\right\}\right)\\
&\geq &\mathrm{Pr}\left(\mathrm{max}_{a\in[n]}\frac{2\langle\bW,\Phi*(\mathbf{A}_{y^{(a)}}-\mathbf{A}_{y}) \rangle}{\|\Phi*(\bA_{y^{(a)}}-\bA_y)\|^2}>1\right)\\
&\geq &\mathrm{Pr}\left(\mathrm{max}_{a\in H}\frac{2\langle \bW,\Phi*(\mathbf{A}_{y^{(a)}}-\mathbf{A}_{y}) \rangle}{\|\Phi*(\bA_{y^{(a)}}-\bA_y)\|^2}>1\right)
\end{eqnarray*}

Let $(\mathcal{X},\mathcal{Y},\mathcal{Z})$ be a partition of $\cup_{s=s_1}^{s_2}[n]^s$ defined by
\begin{eqnarray*}
&&\mathcal{X}=\{\alpha=(\alpha_1,\alpha_2,\ldots,\alpha_s)\in\cup_{s=s_1}^{s_2} [n]^s, \{\alpha_1,\ldots,\alpha_s\}\cap H=\emptyset\}\\
&&\mathcal{Y}=\{\alpha=(\alpha_1,\alpha_2,\ldots,\alpha_s)\in\cup_{s=s_1}^{s_2} [n]^s, |\{i\in[s]: \alpha_i\in H|=1\}\\
&&\mathcal{Z}=\{\alpha=(\alpha_1,\alpha_2,\ldots,\alpha_s)\in \cup_{s=s_1}^{s_2}[n]^s, |\{i\in[s]: \alpha_i\in H|\geq 2\}
\end{eqnarray*}
For $\eta\in\{\mathcal{X},\mathcal{Y},\mathcal{Z}\}$, define the random tensor $\mathbf{W}_{\eta}$ from the entries of $\mathbf{W}$ as follows
\begin{eqnarray*}
(\mathbf{W}_{\eta})_{(i_1,i_2,\ldots,i_s)}=\begin{cases}0&\mathrm{if}\ (i_1,\ldots,i_s)\notin \eta\\ (\mathbf{W})_{(i_1,\ldots,i_s)},&\mathrm{if}\ (i_1,\ldots,i_s)\in \eta\end{cases}
\end{eqnarray*}
 For each $a\in H$, let
\begin{eqnarray*}
&&\mathcal{X}_{a}=\langle\mathbf{W}_{\mathcal{X}},\Phi*(\mathbf{A}_{y^{(a)}}-\mathbf{A}_{y}) \rangle\\
&&\mathcal{Y}_{a}=\langle\mathbf{W}_{\mathcal{Y}},\Phi*(\mathbf{A}_{y^{(a)}}-\mathbf{A}_{y}) \rangle\\
&&\mathcal{Z}_{a}=\langle\mathbf{W}_{\mathcal{Z}},\Phi*(\mathbf{A}_{y^{(a)}}-\mathbf{A}_{y}) \rangle
\end{eqnarray*}
For $s\in\{s_1,s_1+1,\ldots,s_2\}$, let
\begin{eqnarray*}
J_s:=(j_1,\ldots,j_s)\subset [n]^s
\end{eqnarray*}
Explicit computations show that
\begin{eqnarray}
&&(\mathbf{A}_{y^{(a)}})_{J_s}-(\mathbf{A}_{y})_{J_s}\label{ama}\\
&=&\begin{cases}\phi(y^{(a)}(j_1),\ldots,y^{(a)}(j_s))-\phi(y(j_1),\ldots,y(j_s))&\mathrm{if}\ a\in\{j_1,\ldots,j_s\}\\0 &\mathrm{otherwise}.\end{cases}\notag
\end{eqnarray}
\begin{claim}\label{c19}The followings are true:
\begin{enumerate}
\item $\mathcal{X}_{a}=0$ for $a\in H$.
\item For each $a\in H$, the variables $\mathcal{Y}_{a}$ and $\mathcal{Z}_{a}$ are independent.
\end{enumerate}
\end{claim}

\begin{proof}It is straightforward to check (1). (2) holds because $\mathcal{Y}\cap\mathcal{Z}=\emptyset$.
\end{proof}

For $g\in H$, let $\mathcal{Y}^{g}\subseteq \mathcal{Y}$ be defined by
\begin{eqnarray*}
\mathcal{Y}^g=\{\alpha=(\alpha_1,\alpha_2,\ldots,\alpha_s)\in \mathcal{Y}: s\in\{s_1,s_1+1,\ldots,s_2\},\ \exists l\in[s],\ \mathrm{s.t.}\ \alpha_l=g\}.
\end{eqnarray*}
Note that for $g_1,g_2\in H$ and $g_1\neq g_2$, $\mathcal{Y}^{g_1}\cap \mathcal{Y}^{g_2}=\emptyset$. Moreover, $\mathcal{Y}=\cup_{g\in H}\mathcal{Y}^g$. Therefore
\begin{eqnarray*}
\mathcal{Y}_{a}=\sum_{g\in H}\langle\mathbf{W}_{\mathcal{Y}^g},\Phi*(\mathbf{A}_{y^{(a)}}-\mathbf{A}_{y}) \rangle
\end{eqnarray*}
Note also that $\langle\mathbf{W}_{\mathcal{Y}^g},\Phi*(\mathbf{A}_{y^{(a)}}-\mathbf{A}_{y}) \rangle=0$, if $g\neq a$. Hence
\begin{eqnarray*}
\mathcal{Y}_{a}=\sum_{\alpha\in\mathcal{Y}^a}\frac{(\mathbf{W})_{\alpha}\cdot\{(\mathbf{A}_{y^{(a)}}-\mathbf{A}_{y})_{\alpha}\}}{\sigma_{\alpha}}
\end{eqnarray*}

So by (\ref{ama}) we obtain, 
\begin{eqnarray*}
&&\sum_{\alpha\in \mathcal{Y}^a}\frac{(\mathbf{W})_{\alpha}\cdot\{(\mathbf{A}_{y^{(a)}}-\mathbf{A}_{y})_{\alpha}\}}{\si_{\alpha}}
=\sum_{s=s_1}^{s_2}\sum_{j=1}^{s}\sum_{(i_1,\ldots,\widehat{i}_j,\ldots,i_s)\in([n]\setminus H)^{s-1}}\frac{1}{\sigma_{(i_1,\ldots,i_{j-1},a,i_{j+1},\ldots,i_s)}}\\
&&\left\{(\phi(y(i_1),\ldots,y^{(a)}(a),\ldots,y(i_s))-\phi(y(i_1),\ldots,y(a),\ldots,y(i_s)))(\bW)_{(i_1,\ldots,i_{j-1},a,i_{j+1},\ldots,i_s)}\right\}
\end{eqnarray*}
Then $\{\mathcal{Y}_g\}_{g\in H}$ is a collection of independent centered Gaussian random variables. Moreover, the variance of $\mathcal{Y}_g$ is equal to 
\begin{eqnarray}
&&\sum_{s=s_1}^{s_2}\sum_{j=1}^{s}\sum_{(i_1,\ldots,\widehat{i}_j,\ldots,i_s)\in([n]\setminus H)^{s-1}}\frac{1}{\sigma_{(i_1,\ldots,i_{j-1},g,i_{j+1},\ldots,i_s)}^2}\label{vyg}\\
&&(\phi(y(i_1),\ldots,y^{(g)}(g),\ldots,y(i_s))-\phi(y(i_1),\ldots,y(g),\ldots,y(i_s)))^2\notag\\
&=&(1+o(1))L_{\Sigma}(y^{(g)},y)\notag
\end{eqnarray}
by Assumption (3) of the proposition.

By Claim \ref{c19}, we obtain
\begin{eqnarray*}
\frac{2\langle\mathbf{W},\Phi*(\mathbf{A}_{y^{(a)}}-\mathbf{A}_{y}) \rangle}{\|\Phi*(\bA_y^{(a)}-\bA_y)\|^2}=\frac{2\mathcal{Y}_a}{\|\Phi*(\bA_{y^{(a)}}-\bA_y)\|^2}+\frac{2\mathcal{Z}_{a}}{\|\Phi*(\bA_{y^{(a)}}-\bA_y)\|^2}
\end{eqnarray*}
Moreover, 
\begin{eqnarray*}
\max_{a\in H}\frac{2(\mathcal{Y}_a+\mathcal{Z}_{a})}{\|\Phi*(\bA_{y^{(a)}}-\bA_{y})\|^2}&\geq& \max_{a\in H}\frac{ 2\mathcal{Y}_a}{\|\Phi*(\bA_{y^{(a)}}-\bA_{y})\|^2}-\max_{a\in H}\frac{-2\mathcal{Z}_{a}}{\|\Phi*(\bA_{y^{(a)}}-\bA_{y})\|^2}
\end{eqnarray*}

Recall that 
\begin{eqnarray*}
\|\Phi*(\bA_{y^{(a)}}-\bA_y)\|^2=L_{\Phi}(y^{(a)},y)
\end{eqnarray*}

By  Lemma \ref{mg} about the tail bound result of the maximum of Gaussian random variables, if (\ref{epn}) holds with $N$ replaced by $h$, the event
\begin{eqnarray*}
E_1:=\left\{\max_{a\in H}\frac{2\mathcal{Y}_a}{\|\Phi*(\bA_{y^{(a)}}-\bA_{y})\|^2}\geq (1-\epsilon)\sqrt{2\min_{a\in H}\frac{4}{L_{\Phi}(y^{(a)},y)}\log h}\right\}
\end{eqnarray*}
has probability at least $1-e^{-h^{\epsilon}}$; and the event
\begin{eqnarray*}
E_2:=\left\{\max_{a\in H}\frac{2\mathcal{Z}_{a}}{\|\Phi*(\bA_{y^{(a)}}-\bA_y)\|^2}\leq (1+\epsilon)\sqrt{2\log h\cdot \max_{a\in H} \frac{4\mathrm{Var}(\mathcal{Z}_{a})}{(L_{\Phi}(y^{(a)},y))^2}}\right\}
\end{eqnarray*}
has probability $1-h^{-\epsilon}$. 

 Moreover, by Assumption (3) of the Proposition and (\ref{vyg}),
\begin{eqnarray*}
\mathrm{Var} \mathcal{Z}_{a}&=&\|\Phi*(\mathbf{A}_{y^{(a)}}-\mathbf{A}_{y})\|^2-\mathrm{Var}(\mathcal{Y}_a)=o(1)L_{\Phi}(y^{(a)},y)
\end{eqnarray*}
Define an event $E$ by 
\begin{small}
\begin{eqnarray*}
&&E:=\left\{\max_{a\in H}\frac{2\mathcal{Y}_a+2\mathcal{Z}_{a}}{\|\Phi*(\mathbf{A}_{y^{(a)}}-\bA_{y})\|^2}\geq\left(1-\epsilon-(1+\epsilon)o(1)\sqrt{\frac{\max_{a\in H}L_{\Phi}(y^{(a)},y)}{\min_{a\in H}L_{\Phi}(y^{(a)},y)}}\right)\right.\\
&&\left.\times\sqrt{8\log h\min_{a\in H}\frac{1}{L_{\Phi}(y^{(a)},y)}}\right\}
\end{eqnarray*}
\end{small}
Then $E_1\cap E_2\subseteq E$. 

When $n$ is large, and (\ref{as31}) holds
\begin{eqnarray*}
&&\mathrm{Pr}\left(\mathrm{max}_{a\in H}\frac{2\langle\mathbf{W},\Phi*(\mathbf{A}_{y^{(a)}}-\mathbf{A}_{y}) \rangle}{\|\Phi*(\bA_{y^{(a)}}-\bA_y)\|^2}>1\right)\\
&&\geq \mathrm{Pr}(E)\geq \Pr(E_1\cap E_2)\geq 1-\Pr(E_1^c)-\Pr(E_2^c)\rightarrow 1,
\end{eqnarray*}
as $n\rightarrow\infty$. Then the proposition follows.
$\hfill\Box$

\subsection{Examples of Assumption \ref{ap24}}\label{ep24}

We shall see some examples of the function $\theta: \Omega\times [k]\times[k]\rightarrow \RR$ satisfying assumption \ref{ap24}. We first see an example when $\theta$ can uniquely determine the community assignment mapping in $\Omega$. 

\begin{example}Assume $p=k$. For $a,b\in[k]$, $x\in\Omega$
\begin{eqnarray*}
\theta(x,a,b)=\begin{cases}1&\mathrm{if}\ a=b\\ 0&\mathrm{otherwise}.\end{cases}
\end{eqnarray*} 
Then if for all $a\in[p]$ and $j\in[n]$, (\ref{sxze}) holds, we have $x(j)=a$ if and only if $z(j)=a$, then $x=z$.
\end{example}

We now see an example when $\theta$ cannot uniquely determine the community assignment mapping, but determines the community assignment mappings up to the equivalent class as defined in Definition \ref{dfeq}. 

\begin{example}Assume $p=n$. 
\begin{eqnarray*}
\theta(x,i,a)=\begin{cases}1&\mathrm{if}\ x(i)=a;\\ 0&\mathrm{otherwise}.\end{cases}
\end{eqnarray*} 
Then if for all $i,j\in[n]$, (\ref{sxze}) holds, we have $x(j)=x(i)$ if and only if $z(j)=z(i)$, then $x\in C(z)$.
\end{example}

\begin{example}Assume $p=n$, $a\in[k]$ and $i\in[n]$. 
\begin{eqnarray*}
\theta(x,i,a)=x(i)-a;
\end{eqnarray*} 
Then if for all $i,j\in[n]$, (\ref{sxze}) holds, we have $x(i)-x(j)=z(i)-z(j)$. This implies that $x(i)=x(j)$ if and only if $z(i)=z(j)$, therefore $x\in C(z)$. If both $x$ and $z$ are surjective onto $[k]$, then $x=z$.
\end{example}

\subsection{Example of Theorem \ref{p31}}\label{at41}
\begin{example}\label{ep45}Here we see an example about how to apply Theorem \ref{p31} to the exact recovery of community detection on hypergraphs. Let $y\in\Omega_{n_1,\ldots,n_k}$ be the true community assignment mapping. Assume that for any $s\in\{s_1,\ldots,s_2\}$, $(i_1,i_2,\ldots,i_s), (j_1,j_2,\ldots,j_s)\in[n]^s$, we have
\begin{eqnarray*}
\sigma_{(i_1,i_2,\ldots,i_s)}=\sigma_{(j_1,j_2,\ldots,j_s)}
\end{eqnarray*}
whenever
\begin{eqnarray*}
y(i_r)=y(j_r),\ \forall r\in[s];
\end{eqnarray*}
i.e., $\sigma_{(i_1,\ldots,i_s)}$ depends only on the communities of $(i_1,\ldots,i_s)$ under the mapping $y$. In this case we can define $\overline{\sigma}: \cup_{s=s_1}^{s_2}[k]^s\rightarrow (0,\infty)$, such that
\begin{eqnarray}
\sigma_{(i_1,\ldots,i_s)}=\overline{\sigma}(y(i_1),\ldots,y(i_s)),\ \forall (i_1,\ldots,i_s)\in[n]^s\label{dfsb}
\end{eqnarray}

 Then for any $a\in[n]$,
\begin{eqnarray}
&&L_{\Phi}(y^{(a)},y)=\|\Phi*(\bA_{y^{(a)}}-\bA_y)\|^2\label{lsay}\\
&=&\sum_{s=s_1}^{s_2}\sum_{(i_1,\ldots,i_s)\in[n]^s}\frac{(\phi(y^{(a)}(i_1),\ldots,y^{(a)}(i_s))-\phi(y(i_1),\ldots,y(i_s)))^2}{(\ol{\sigma}(y(i_1),\ldots,y(i_s)))^2}\notag
\end{eqnarray}
Moreover, for any $a,b\in[n]$ such that 
\begin{eqnarray*}
y(a)=y(b);\ y^{(a)}(a)=y^{(b)}(b)
\end{eqnarray*}
we have
\begin{eqnarray*}
L_{\Phi}(y^{(a)},y)=L_{\Phi}(y^{(b)},y).
\end{eqnarray*}
We consider
\begin{eqnarray}
\min_{(y^{(a)}(a),y(a))\in [k]^2,y^{(a)}(a)\neq y(a)}L_{\Phi}(y^{(a)},y)\label{ddt}
\end{eqnarray} 
Assume that when $y(a)=r_0$, $y^{(a)}(a)=r_1$, $L_{\Phi}(y^{(a)},y)$ achieves its minimum. Let $H\subset y^{-1}(r_0)$, then $h=|H|\leq n_{r_0}$. Assume
\begin{eqnarray*}
\lim_{n\rightarrow\infty}\frac{\log n_{r_0}}{\log n}=1.
\end{eqnarray*}
Then we may choose $h=\frac{n_{r_0}}{\log n}$ such that Assumptions (1)(2) in Theorem \ref{p31} hold. Moreover, Assumption (4) in Theorem \ref{p31} holds because if for all $a\in H$, let $y^{(a)}(a)=r_1$, then $L_{\Phi}(y^{(a)},y)$ takes the same value for all $a\in H$. There are many mappings $\phi: \cup_{s=s_1}^{s_1}[k]^s\rightarrow \RR$ to guarantee Assumption (3) in Theorem \ref{p31}. For example, one may choose
\begin{eqnarray}
\phi(b_1,\ldots,b_s)=\begin{cases}2^s&\mathrm{if}\ b_1=\ldots=b_s\\0&\mathrm{otherwise}.\end{cases}\label{dph}
\end{eqnarray}
for $s\in\{s_1,s_1+1,\ldots,s_2\}$ and $b_1,\ldots,b_s\in[k]$. Then from (\ref{lsay}) we obtain
\begin{eqnarray}
&&L_{\Phi}(y^{(a)},y)\label{lshg}=\sum_{s=s_1}^{s_2}\sum_{(b_1,\ldots,b_s)\in[k]^s}\sum_{(d_1,\ldots,d_s)\in[k]^s}\\
&&\frac{(\phi(d_1,\ldots,d_s)-\phi(b_1,\ldots,b_s))^2}{(\overline{\sigma}(b_1,\ldots,b_s))^2}\left(\prod_{j=1}^{s}t_{d_j,b_j}(y^{(a)},y)\right)\notag
\end{eqnarray}
From (\ref{t1})-(\ref{t4}) and (\ref{dph}) we obtain that the terms actually contributing to the sum must satisfy
\begin{eqnarray*}
\{(d_1,b_1),\dots,(d_s,b_s)\}\subseteq\{(r_1,r_1),(r_1,r_0)\}
\end{eqnarray*}
or 
\begin{eqnarray*}
\{(d_1,b_1),\dots,(d_s,b_s)\}\subseteq\{(r_0,r_0),(r_1,r_0)\}
\end{eqnarray*}
Then we obtain
\begin{eqnarray*}
L_{\Phi}(y^{(a)},y)&=&\sum_{s=s_1}^{s_2}2^{2s}(L_{0,s}+L_{1,s})
\end{eqnarray*}
where
\begin{eqnarray*}
L_{0,s}&=&\sum_{(b_1,\ldots,b_s)\in[k]^s, (d_1,\ldots,d_s)\in[k]^s, (d_1,b_1),\dots,(d_s,b_s)\subseteq\{(r_0,r_0),(r_1,r_0)\}}\frac{\left(\prod_{j=1}^{s}t_{d_j,b_j}(y^{(a)},y)\right)}{(\overline{\sigma}(r_0,\ldots,r_0))^2}\\
L_{1,s}&=&\sum_{(b_1,\ldots,b_s)\in[k]^s, (d_1,\ldots,d_s)\in[k]^s, (d_1,b_1),\dots,(d_s,b_s)\subseteq\{(r_1,r_1),(r_1,r_0)\}}\frac{\left(\prod_{j=1}^{s}t_{d_j,b_j}(y^{(a)},y)\right)}{(\overline{\sigma}(b_1,\ldots,b_s))^2}
\end{eqnarray*}
Assume 
\begin{eqnarray*}
\lim_{n\rightarrow\infty}\min\{n_{r_0},n_{r_1}\}=\infty.
\end{eqnarray*}
If in $\{(d_1,b_1),\ldots,(d_s,b_s)\}$, there exist more than one $g\in[s]$, such that $(d_g,b_g)=(r_1,r_0)$, then the sum of such terms will be of order $o((n_{r_0})^{s-1})$ (resp.\ $o((n_{r_1})^{s-1})$) in $L_{0,s}$ (resp.\ $L_{1,s}$). Therefore we obtain
\begin{eqnarray*}
L_{0,s}&=&\frac{s\left(n_{r_0}\right)^{s-1}(1+o(1))}{(\overline{\sigma}(r_0,\ldots,r_0))^2}
\end{eqnarray*}
To analyze $L_{1,s}$, assume that there exists a positive constant $C>0$ independent of $n$, such that 
\begin{eqnarray}
0<C<\frac{\min_{b_1,\ldots,b_s\in\{r_0,r_1\}}\overline{\si}(b_1,\ldots,b_s)}{\max_{b_1,\ldots,b_s\in\{r_0,r_1\}}\overline{\si}(b_1,\ldots,b_s)},\ \forall n\in \NN, s\in\{s_1,s_1+1,\ldots s_2\}.\label{qlb}
\end{eqnarray}
Then we obtain
\begin{eqnarray*}
L_{1,s}&=&\sum_{j=1}^{s}\frac{\left(n_{r_1}\right)^{s-1}(1+o(1))}{(\overline{\sigma}(r_1,\ldots,r_1,r_0,r_1,\ldots,r_1))^2}
\end{eqnarray*}
Then Assumption (3) of Theorem \ref{p31} follows from the fact that $|H|=\frac{n_{r_0}}{\log n}=o(n_{r_0})$.

In the special case when all the communities have equal size, we may obtain a sufficient condition that the exact recovery of MLE does not occur in the hypergraph case when the number of communites $k=e^{o(\log n)}$. Since in this case we have $n_1=n_2=\ldots=n_k\geq e^{\log n-o(\log n)}$, then (\ref{nii}) holds. Choose $h=\frac{n_1}{\log n}$, then Assumptions (1) and (2) of Theorem \ref{p31} hold.
\end{example}

\subsection{Example of Theorem \ref{p215}}\label{at42}

\begin{example}We can also apply Theorem \ref{p215} to the case of exact recovery of community detection on hypergraphs. Again we consider the case when $\sigma_{(i_1,\ldots,i_s)}$ depends only on $(y(i_1),\ldots,y(i_s))$. Hence we may define $\overline{\si}$ as in (\ref{dfsb}). 


 To check Assumption \ref{ap214}, let $y_m,y_{m+1},x\in\Omega$ be given as in the proof of Proposition \ref{p215}. For the simplicity of notation, we use $y$ instead of $y^*$. By (\ref{lshg}) we obtain
\begin{eqnarray*}
&&L_{\Phi}(x,y_m)-L_{\Phi}(x,y_{m+1})\\
&=&\sum_{s=s_1}^{s_2}\sum_{(i_1,\ldots,i_s)\in[n]^s}\frac{1}{(\ol{\sigma}(y(i_1),\ldots,y(i_s)))^2}\\
&&\left[(\phi(x(i_1),\ldots,x(i_s))-\phi(y_m(i_1),\ldots,y_m(i_s)))^2-(\phi(x(i_1),\ldots,x(i_s))-\phi(y_{m+1}(i_1),\ldots,y_{m+1}(i_s)))^2\right]\\
&=&\sum_{s=s_1}^{s_2}\sum_{(i_1,\ldots,i_s)\in[n]^s}\frac{1}{(\ol{\sigma}(y(i_1),\ldots,y(i_s)))^2}\left\{(\phi(y_m(i_1),\ldots,y_m(i_s))^2-(\phi(y_{m+1}(i_1),\ldots,y_{m+1}(i_s))^2\right.\\
&&\left.-2\phi(x(i_1),\ldots,x(i_s))\left[\phi(y_m(i_1),\ldots,y_m(i_s))-\phi(y_{m+1}(i_1),\ldots,y_{m+1}(i_s))\right]\right\}
\end{eqnarray*}
For $j,p,q\in[k]$, and $x,y,z\in\Omega$, define
\begin{eqnarray*}
t_{j,p,q}(x,y,z)=|\{i\in[n]: x(i)=j,y(i)=p,z(i=q)\}|=|x^{-1}(j)\cap y^{-1}(p)\cap z^{-1}(q)|.
\end{eqnarray*}
Then
\begin{eqnarray}
&&L_{\Phi}(x,y_m)-L_{\Phi}(x,y_{m+1})=\sum_{s=s_1}^{s_2}\sum_{(b_1,\ldots,b_s)\in[k]^s}\frac{1}{(\ol{\si}(b_1,\ldots,b_s))^2}\sum_{(d_1,\ldots,d_s)\in[k]^s}\label{sls}\\
&&\left\{(\phi(d_1,\ldots,d_s))^2 \left(\prod_{r=1}^{s}t_{b_r,d_r}(y,y_m)-\prod_{r=1}^{s}t_{b_r,d_r}(y,y_{m+1})\right)-2\sum_{(l_1,\ldots,l_s)\in[k]^s}\right.\notag\\
&&\left.\phi(l_1,\ldots,l_s)\phi(d_1,\ldots,d_s) \left(\prod_{r=1}^{s}t_{b_r,d_r,l_r}(y,y_m,x)-\prod_{r=1}^{s}t_{b_r,d_r,l_r}(y,y_{m+1},x)\right)\right\}\notag
\end{eqnarray}
Recall that $D_{\Omega}(y_m,y_{m+1})=1$, and there exists $u\in [n]$ such that 
\begin{eqnarray*}
x(u)=j=y_{m+1}(u)\neq y_m (u)=i=y(u).
\end{eqnarray*}
 where $i,j\in[k]$ and $i\neq j$; while $y_m(v)=y_{m+1}(v)$ for all the $v\in[n]\setminus \{u\}$. This implies that if $\{d_1,\ldots,d_s\}\cap \{i,j\}=\emptyset$, then the corresponding summand in (\ref{sls}) is 0 and does not contribute to the sum. Under the assumption that  
 \begin{enumerate}
 \item $(t_{1,1}(x,y),t_{1,2}(x,y),\ldots,t_{k,k}(x,y))\in \mathcal{B}_{\epsilon}$
with $w:[k]\rightarrow[k]$ the identity map; and
\item $n_1\geq n_2\geq\ldots\geq n_k$; and
\item $\min_{(b_1,\ldots,b_s)\in[k]^s}|\ol{\si}(b_1,\ldots,b_s)|\geq B_3>0$; and
\item $\lim_{n\rightarrow\infty}\frac{n\epsilon}{n_1}=0$.
\end{enumerate}
we obtain
\begin{eqnarray*}
&&L_{\Phi}(x,y_m)-L_{\Phi}(x,y_{m+1})=\sum_{s=s_1}^{s_2}\sum_{g=1}^{s}\sum_{(b_1,\ldots,\widehat{b}_g,\ldots,b_s)\in[k]^s}\frac{1}{(\ol{\si}(b_1,\ldots,i,\ldots,b_s))^2}\sum_{(d_1,\ldots,\widehat{d}_g,\ldots,d_s)\in[k]^s}\\
&&\left\{(\phi(d_1,\ldots,i,\ldots,d_s))^2-(\phi(d_1,\ldots,j,\ldots,d_s))^2) \prod_{r\in[s]\setminus\{g\}}t_{b_r,d_r}(y,y_m)-2\sum_{(l_1,\ldots,\widehat{l}_g,\ldots,l_s)\in[k]^s}\right.\\
&&\phi(l_1,\ldots,j,\ldots,l_s)\left(\phi(d_1,\ldots,i,\ldots,d_s)-\phi(d_1,\ldots,j,\ldots,d_s) \right)\\
&&\left.\left(\prod_{r\in[s]\setminus\{g\}}t_{b_r,d_r,l_r}(y,y_m,x)\right)\right\}+O\left(\frac{n_1^{k-2}}{B_3^2}\right)
\end{eqnarray*}
The identity above can be interpreted as follows. We can classify the terms satisfying $\{d_1,\ldots,d_s\}\cap \{i,j\}\neq \emptyset$ by the number 
\begin{eqnarray*}
N_{i,j}=\{l\in[s]: d_l\in\{i,j\}\},
\end{eqnarray*}
and obtain that the leading term of $L_{\Phi}(x,y_m)-L_{\Phi}(x,y_{m+1})$ is given by the terms when $N_{i,j}=1$. Moreover, by Assumption (1) we have
\begin{eqnarray*}
&&L_{\Phi}(x,y_m)-L_{\Phi}(x,y_{m+1})=\sum_{s=s_1}^{s_2}\sum_{g=1}^{s}\sum_{(b_1,\ldots,\widehat{b}_g,\ldots,b_s)\in[k]^s}\frac{1}{(\ol{\si}(b_1,\ldots,i,\ldots,b_s))^2}\\
&&\left\{(\phi(b_1,\ldots,i,\ldots,b_s))^2-(\phi(b_1,\ldots,j,\ldots,b_s))^2) \prod_{r\in[s]\setminus\{g\}}n_r-2\sum_{(l_1,\ldots,\widehat{l}_g,\ldots,l_s)\in[k]^s}\right.\\
&&\phi(b_1,\ldots,j,\ldots,b_s)\left(\phi(b_1,\ldots,i,\ldots,b_s)-\phi(b_1,\ldots,j,\ldots,b_s) \right)\\
&&\left.\left(\prod_{r\in[s]\setminus\{g\}}n_r\right)\right\}+O\left(\frac{n_1^{s-2}}{B_3^2}\right)+O\left(\frac{\epsilon n n_1^{s-2}}{B_3^2}\right)\\
&&=\sum_{s=s_1}^{s_2}\sum_{g=1}^{s}\sum_{(b_1,\ldots,\widehat{b}_g,\ldots,b_s)\in[k]^s}\frac{1}{(\ol{\si}(b_1,\ldots,i,\ldots,b_s))^2}\\
&&(\phi(b_1,\ldots,i,\ldots,b_s))-(\phi(b_1,\ldots,j,\ldots,b_s)))^2 \prod_{r\in[s]\setminus\{g\}}n_{b_r}\\
&&+O\left(\frac{n_1^{s-2}}{B_3^2}\right)+O\left(\frac{\epsilon n n_1^{s-2}}{B_3^2}\right)
\end{eqnarray*} 
Define
 \begin{eqnarray*}
\Delta:&=&\min_{i,j\in[k],i\neq j}\sum_{s=s_1}^{s_2}\sum_{g=1}^{s}\sum_{(b_1,\ldots,\widehat{b}_g,\ldots,b_s)\in[k]^s}\frac{1}{(\ol{\si}(b_1,\ldots,i,\ldots,b_s))^2}\\
&&\times(\phi(b_1,\ldots,i,\ldots,b_s))-(\phi(b_1,\ldots,j,\ldots,b_s)))^2 \prod_{r\in[s]\setminus\{g\}}n_{b_r}
\end{eqnarray*}
We further make the assumptions below:
\begin{eqnarray}
\lim_{n\rightarrow\infty}\frac{n_1^{s-2}+\epsilon  n n_1^{s-2}}{B_3^2\Delta}=0.\label{eec1}
\end{eqnarray}
Then
\begin{eqnarray*}
L_{\Phi}(x,y_m)-L_{\Phi}(x,y_{m+1})
&\geq &\Delta(1+o(1))
\end{eqnarray*}

Then by Assumption \ref{ap27}(1), the exact recovery occurs with probability 1 when $n\rightarrow \infty$ if (\ref{ld1}), (\ref{ld2}) hold, and 
\begin{eqnarray}
&&\sum_{s=s_1}^{s_2}\sum_{(i_1,\ldots,i_s)\in [k]^s}\sum_{(j_1,\ldots,j_s\in[k]^s)}\left(\phi(i_1,\ldots,i_s)-\phi(j_1,\ldots,j_s)\right)^2\label{esc2}\\
&&\times\left(\prod_{r=1}^s t_{i_r,j_r}(x,y)\right)\geq T(n)\notag
\end{eqnarray}
when (\ref{tcd}) holds.

There are a lot of functions $\phi:\cup_{s=s_1}^{s_2}[k]^s\rightarrow\RR$ satisfying (\ref{eec1}), (\ref{esc2})
and Assumption \ref{ap24}. For example, we may consider the function $\phi$ as defined in (\ref{dph}). Assume 
\begin{eqnarray*}
\Delta&=&\sum_{s=s_1}^{s_2}\sum_{g=1}^{s}\sum_{(b_1,\ldots,\widehat{b}_g,\ldots,b_s)\in[k]^s}\frac{1}{(\ol{\si}(b_1,\ldots,r_0,\ldots,b_s))^2}\\
&&\times(\phi(b_1,\ldots,r_0,\ldots,b_s))-(\phi(b_1,\ldots,r_1,\ldots,b_s)))^2 \prod_{r\in[s]\setminus\{g\}}n_{b_r}
\end{eqnarray*}
where $r_0,r_1\in[k]$ and $r_0\neq r_1$. As in Example \ref{ep45}, we obtain that
\begin{eqnarray*}
\Delta=\sum_{s=s_1}^{s_2}2^{2s}(\Delta_{0,s}+\Delta_{1,s}),
\end{eqnarray*}
where 
\begin{eqnarray*}
\Delta_{0,s}&=&\frac{s(n_{r_0})^{s-1}}{(\ol{\si}(r_0,\ldots,r_0))^2}\geq \frac{s((n_{r_0}))^{s-1}}{B_1^2}\\
\Delta_{1,s}&=&\sum_{j=1}^s\frac{(n_{r_1})^{s-1}}{(\ol{\si}(r_1,\ldots,r_1,r_0,r_1,\ldots,r_1))^2}
\geq \frac{s((n_{r_1}))^{s-1}}{B_1^2}
\end{eqnarray*}
where the inequality follows from Assumption \ref{ap27}(2). It is straightforward to check that when $\phi$ is given by (\ref{dph}), (\ref{eec1}) holds if (\ref{qlb}) holds.

To check (\ref{esc2}), note that
\begin{eqnarray*}
&&\sum_{s=s_1}^{s_2}\sum_{(i_1,\ldots,i_s)\in [k]^s}\sum_{(j_1,\ldots,j_s\in[k]^s)}\left(\phi(i_1,\ldots,i_s)-\phi(j_1,\ldots,j_s)\right)^2\times\left(\prod_{r=1}^s t_{i_s,j_s}(x,y)\right)\\
&\geq & \sum_{s=s_1}^{s_2} \sum_{g\in[s]}\sum_{j\in [s]}\sum_{i\in[k],i\neq j}
(\phi(w(i),\ldots,w(i))-\phi(i,\ldots,i,j,i,\ldots,i))^2\\
&&\times t_{w(i),j}(x,y)\prod_{r=[s]\setminus\{g\}} t_{w(i),i}(x,y)
\end{eqnarray*}
When (\ref{tcd})holds, the following cases might occur
\begin{itemize}
    \item $w$ is not a bijection from $[k]$ to $[k]$. In this case, there exists $i,j\in[k]$, such that $w(i)=w(j)$, then when (\ref{tcd}) holds, we obtain
    \begin{eqnarray*}
    t_{w(i),j}=t_{w(j),j}\geq \frac{n_{j}}{k}
    \end{eqnarray*}
    \item $w$ is a bijection from $[k]$ to $[k]$. However, there exists $i\in[k]^2$, such that
    \begin{eqnarray*}
    t_{w(j),j}\leq n_i-\epsilon n.
    \end{eqnarray*}
    Let
    \begin{eqnarray*}
    i:=w^{-1}(\mathrm{argmax}_{l\in[k]\setminus\{w(j)\}}t_{l,j}),
    \end{eqnarray*}
    then $i\neq j$ and 
    \begin{eqnarray*}
    t_{w(i),j}(x,y)\geq \frac{\epsilon n}{k-1}
    \end{eqnarray*}
\end{itemize}
When (\ref{tcd}) holds and $y\in \Omega_c$ we have
\begin{eqnarray*}
&&\sum_{s=s_1}^{s_2}\sum_{(i_1,\ldots,i_s)\in [k]^s}\sum_{(j_1,\ldots,j_s\in[k]^s)}\left(\phi(i_1,\ldots,i_s)-\phi(j_1,\ldots,j_s)\right)^2\times\left(\prod_{r=1}^s t_{i_s,j_s}(x,y)\right)\\
&\geq &\sum_{s=s_1}^{s_2} s 2^{2s}\left(\frac{n_i}{k}\right)^{s-1}\min\left\{\frac{n_j}{k},\frac{\epsilon n}{k-1}\right\}
\geq \sum_{s=s_1}^{s_2} \frac{2^{2s}sn^s}{(ck)^{s-1}\max\left\{ck,\frac{k-1}{\epsilon}\right\}}
\end{eqnarray*}
Let 
\begin{eqnarray*}
T(n):=\sum_{s=s_1}^{s_2} \frac{2^{2s}sn^s}{(ck)^{s-1}\max\left\{ck,\frac{k-1}{\epsilon}\right\}}
\end{eqnarray*}
Then we obtain (\ref{esc2}). 
\end{example}

\section{Community Detection on Gaussian Mixture Models with Fixed Number of Vertices in Each Community}\label{gm2}

In this section, we consider the MLE restricted to the sample space consisting of all the mappings satisfying the condition that the number of vertices in each community is the same as that of the true community assignment mapping $y\in \Omega_{n_1,\ldots,n_k}$. Again we shall prove a sufficient condition for the occurrences of exact recovery.

Let $x\in\Omega_{n_1,\ldots,n_k}$. By (\ref{kymx}),
\begin{eqnarray*}
\check{y}:=\mathrm{argmin}_{x\in \Omega_{n_1,\ldots,n_k}}\|\Phi*(\bK_y-\bA_x)\|^2=\mathrm{argmin}_{x\in \Omega_{n_1,\ldots,n_k}}f(x)
\end{eqnarray*}

Recall that $f(x)$ is defined as in (\ref{df}). 
Recall also that $f(x)-f(y)$ is a Gaussian random variable with mean value $L_{\Phi}(x,y)$ and variance $4L_{\Phi}(x,y)$.

For each $x\in \Omega_{n_1,\ldots,n_k}$, let
\begin{eqnarray*}
C^*(x):=C(x)\cap \Omega_{n_1,\ldots,n_k};
\end{eqnarray*}
i.e. $C^*(x)$ consists of all the community assignment mappings in $\Omega_{n_1,\ldots,n_k}$ that are equivalent to $x$ in the sense of Definition \ref{dfeq}.
Let
\begin{eqnarray*}
\ol{\Omega}_{n_1,\ldots,n_k}:=\{C^*(x):x\in\Omega_{n_1,\ldots,n_k}\};
\end{eqnarray*}
i.e. $\ol{\Omega}_{n_1,\ldots,n_k}$ consists of all the equivalence classes in $\Omega_{n_1,\ldots,n_k}$.

\begin{lemma}For $x,z\in \Omega_{n_1,\ldots,n_k}$. If $x\in C^*(z)$, then
\begin{eqnarray*}
f(x)=f(z).
\end{eqnarray*}
\end{lemma}
\begin{proof}The lemma follows from Lemma \ref{lfe}.
\end{proof}

Define
\begin{eqnarray*}
p(\check{y};\sigma):=\Pr(\check{y}\in C(y))=\Pr\left(f(\check{y})<\min_{C^*(x)\in \left(\ol{\Omega}_{n_1,\ldots,n_k}\setminus \{C^*(y)\}\right)}f(x)\right)
\end{eqnarray*}
Then
\begin{eqnarray}
1-p(\check{y};\sigma)&\leq& \sum_{C^*(x)\in \left(\ol{\Omega}_{n_1,\ldots,n_k}\setminus \{C^*(y)\}\right)}\Pr(f(x)-f(y)\leq 0)\label{wmc}\\
&=&\sum_{C^*(x)\in \left(\ol{\Omega}_{n_1,\ldots,n_k}\setminus \{C^*(y)\}\right)}\Pr_{\xi\in \sN(0,1)}\left(\xi\geq \frac{L_{\Phi}(x,y)}{2}\right)\notag\\
&\leq&\sum_{C^*(x)\in \left(\ol{\Omega}_{n_1,\ldots,n_k}\setminus \{C^*(y)\}\right)} e^{-\frac{L_{\Phi}(x,y)}{8}}\notag
\end{eqnarray}

\begin{lemma}\label{ll65}Let $y\in \Omega_{n_1,\ldots,n_k}\cap \Omega_c$ be the true community assignment mapping.  Let $x\in\Omega_{n_1,\ldots,n_k}$ 
For $i\in [k]$,  let $w(i)\in[k]$ be defined as in (\ref{twi}). Then
\begin{enumerate}
\item when $\epsilon\in\left(0,\frac{c}{k}\right)$ and $(t_{1,1}(x,y),\ldots, t_{k,k}(x,y))\in \mathcal{B}_{\epsilon}$, $w$ is a bijection from $[k]$ to $[k]$.
\item Assume there exist $i,j\in[k]$, such that $n_i\neq n_j$. If
\begin{eqnarray}
\epsilon<\min_{i,j\in[k]:n_i\neq n_j}\left|\frac{n_i-n_j}{n}\right|\label{eu1}
\end{eqnarray}
Then for any $i\in[k]$,
\begin{eqnarray}
n_i=|y^{-1}(i)|=|y^{-1}(w(i))|=n_{w(i)}.\label{esc}
\end{eqnarray}
\end{enumerate}
\end{lemma}

\begin{proof}See Lemma 6.6 of \cite{ZL19}.
\end{proof}

\begin{definition}Let $l\geq 2$ be a positive integer. Let $x,y\in \Omega_{n_1,\ldots,n_k}$. We say
$l$ distinct communities $(i_1,\ldots,i_{l})\in[k]^l$ is an $l$-cycle for $(x,y)$, if $t_{i_{s-1},i_s}(x,y)>0$ for all $2\leq s\leq l+1$, where $i_{l+1}:=i_1$.
\end{definition}

\begin{lemma}\label{lm33}Let $x,y\in \Omega_{n_1,\ldots,n_k}$ and $x\neq y$. Then there exists an $l$-cycle for $(x,y)$ with $2\leq l\leq k$. 
\end{lemma}

\begin{proof}See Lemma 3.3 of \cite{ZL19}.
\end{proof}

\begin{lemma}\label{lm63}For any $x,y\in\Omega_{n_1,\ldots,n_k}$, $L_{\Phi}(x,y)\geq 0$, where the equality holds if and only if $x\in C^*(y)$.
\end{lemma}

\begin{proof}The lemma follows from Lemma \ref{lgz}.
\end{proof}

\begin{lemma}\label{lm46}Suppose that Assumption \ref{ap214} holds.  Then Assumption \ref{ap46} holds.
\end{lemma}

\begin{proof}Let $z_0=y_m$ and $z_j=y_h$. For $i\in[j-1]$, define $z_i\in \Omega$ by
\begin{eqnarray*}
z_i(v)=\begin{cases}z_{i-1}(v)&\mathrm{if}\ v\in[n]\setminus \{u_i\}\\ x(u_i)&\mathrm{if}\ v=u_i \end{cases}
\end{eqnarray*}
Then for any $i\in [j]$,
\begin{eqnarray*}
D_{\Omega}(z_i,z_{i-1})=1.
\end{eqnarray*}
 by Assumption \ref{ap214}, we obtain
\begin{eqnarray}
L_{\Phi}(x,z_{i-1})- L_{\Phi}(x,z_i)\geq \Delta(1+o(1)),\ \forall i\in [j]\label{itd}
\end{eqnarray}
summing over all the $i\in [j]$, we obtain (\ref{l46d}).
\end{proof}

\bigskip

\noindent{\textbf{Proof of Theorem \ref{m27}.}}
Let 
\begin{eqnarray*}
\Gamma:=\sum_{C^*(x)\in \left(\ol{\Omega}_{n_1,\ldots,n_k}\setminus \{C^*(y)\}\right)}e^{-\frac{L_{\Phi}(x,y)}{8}}.
\end{eqnarray*}
By (\ref{wmc}), it suffices to show that $\lim_{n\rightarrow\infty}\Gamma=0$.

Let
\begin{eqnarray*}
0<\epsilon<\min\left(\frac{2c}{3k},\min_{i,j\in[k],n_i\neq n_j}\left|\frac{n_i-n_j}{n}\right|\right)
\end{eqnarray*}
Note that
\begin{eqnarray*}
\Gamma\leq \Gamma_1+\Gamma_2;
\end{eqnarray*}
where 
\begin{eqnarray*}
\Gamma_1&=&\sum_{C^*(x)\in\ol{\Omega}_{n_1,\ldots,n_k}:\left(t_{1,1}(x,y),\ldots,t_{k,k}(x,y)\right)\in\left(\mathcal{B}\setminus \mathcal{B}_{\epsilon}\right), C^*(x)\neq C^*(y)}e^{\frac{-L_{\Phi}(x,y)}{8}}
\end{eqnarray*}
and
\begin{eqnarray}
\Gamma_2=\sum_{C^*(x)\in\ol{\Omega}_{n_1,\ldots,n_k}:\left(t_{1,1}(x,y),\ldots, t_{k,k}(x,y)\right)\in \mathcal{B}_{\epsilon}, C(x)\neq C(y)}e^{\frac{-L_{\Phi}(x,y)}{8}}.\label{dg2}
\end{eqnarray} 

Under Assumption \ref{ap27}, by Lemma \ref{l13} we have
\begin{eqnarray*}
0\leq \Gamma_1&\leq& k^n e^{-\frac{T(n)}{B_1^2}}
\end{eqnarray*}
By (\ref{ld1}), we have
\begin{eqnarray}
\lim_{n\rightarrow\infty}\Gamma_1=0.\label{g1z}
\end{eqnarray}

Now let us consider $\Gamma_2$. Recall that $y\in \Omega_{n_1,\ldots,n_k}\cap \Omega_c$ is the true community assignment mapping.  Let $w$ be the bijection from $[k]$ to $[k]$ as defined in (\ref{twi}). Let $y^*\in \Omega$ be defined by
\begin{eqnarray*}
y^*(z)=w(y(z)),\ \forall z\in [n].
\end{eqnarray*}
Then $y^*\in C(y)$. By Part (2) of Lemma \ref{ll65}, we obtain that for $i\in[k]$
\begin{eqnarray*}
\left|(y^*)^{-1}(i)\right|=\left|y^{-1}(w^{-1}(i))\right|=\left|y^{-1}(i)\right|;
\end{eqnarray*}
therefore $y^*\in \Omega_{n_1,\ldots,n_k}$.
 Moreover, $x$ and $y^*$ satisfies
\begin{eqnarray}
t_{i,i}(x,y^*)\geq n_i(y^*)-n\epsilon,\ \forall i\in[k].\label{dxys1}
\end{eqnarray}

If $x\neq y^*$, by Lemma \ref{lm33},
 there exists an $l$-cycle $(i_1,\ldots,i_{l})$ for $(x,y^*)$ with $2\leq l\leq k$. Then for each $2\leq a\leq (l+1)$, choose an arbitrary vertex $u_m$ in $S_{i_{m-1},i_m}(x,y^*)$, and let $y_1(u_m)=i_{m-1}$, where $i_{l+1}:=i_1$. For any vertex $z\in[n]\setminus\{u_{2},\ldots,u_{l+1}\}$, let $y_1(z)=y^*(z)$.

 Note that $y_1\in \Omega_{n_1,\ldots,n_k}$. Moreover, for $1\leq m\leq l$, we have
\begin{eqnarray*}
t_{i_m,i_m}(x,y^*)+1=t_{i_m,i_m}(x,y_1);\\
t_{i_m,i_{m+1}}(x,y^*)-1=t_{i_m,i_{m+1}}(x,y_1)
\end{eqnarray*}
and
\begin{eqnarray*}
t_{a,b}(x,y^*)=t_{a,b}(x,y_1),\ \forall (a,b)\notin\{(i_m,i_m),(i_m,i_{m+1})\}_{s=1}^{l}.
\end{eqnarray*}
When $\left(t_{1,1}(x,y),\ldots, t_{k,k}(x,y)\right)\in\mathcal{B}_{\epsilon}$, From Assumption \ref{ap214} and Lemma \ref{lm46} we obtain
\begin{eqnarray*}
L_{\Phi}(x,y_1)-L_{\Phi}(x,y)&=&L_{\Phi}(x,y_1)-L_{\Phi}(x,y^*)\leq -l\Delta(1+o(1))
\end{eqnarray*}
Therefore 
\begin{eqnarray}
e^{-\frac{L_{\Phi}(x,y)}{8}}\leq e^{-\frac{L_{\Phi}(x,y_1)}{8}}e^{-\frac{l\Delta(1+o(1))}{8}}\label{mst1}
\end{eqnarray}

If $y_1\neq x$, we find an $l_2$-cycle ($2\leq l_2\leq k$) for $(x,y_1)$, change community assignments along the $l_2$-cycle  as above, and obtain another community assignment mapping $y_2\in \Omega_{n_1,\ldots,n_k}$, and so on. 
Let $y_0:=y$, and note that for each $r\geq 1$, if $y_r$ is obtained from $y_{r-1}$ by changing colors along an $l_r$ cycle for $(x,y_{r-1})$, we have
\begin{eqnarray*}
D_{\Omega}(x,y_r)= D_{\Omega}(x,y_{r-1})-l_r
\end{eqnarray*}
Therefore finally we can obtain $x$ from $y$ by changing colors along at most $\left\lfloor \frac{n}{2} \right\rfloor$ cycles. By similar arguments as those used to derive (\ref{mst1}), we obtain that for each $r$
\begin{eqnarray*}
e^{-\frac{L_{\Phi}(x,y_{r-1})}{8}}\leq e^{-\frac{L_{\Phi}(x,y_r)}{8}}e^{-\frac{l_r\Delta(1+o(1))}{8}}
\end{eqnarray*}
Therefore if $y_h=x$ for some $1\leq h\leq \left\lfloor \frac{n}{2} \right\rfloor$, we have
\begin{eqnarray*}
e^{-\frac{L_{\Phi}(x,y)}{8}}\leq e^{-\frac{L_{\Phi}(x,y_{h-1})}{8}}e^{-\frac{\left(\sum_{r=1}^{h-1}l_r\right)\Delta(1+o(1)))}{8}}.
\end{eqnarray*}
By Assumption \ref{ap214} and Lemma \ref{lm46} we obtain
\begin{eqnarray*}
L_{\Phi}(x,y_{h-1}))^2\geq l_h\Delta(1+o(1))
\end{eqnarray*}
Therefore
\begin{eqnarray*}
e^{-\frac{L_{\Phi}(x,y)}{8}}\leq \prod_{i\in[h]}e^{-\frac{l_i\Delta(1+o(1))}{8}}.
\end{eqnarray*}

Note also that for any $r_1\neq r_2$, in the process of obtaining $y_{r_1}$ from $y_{r_1-1}$ and the process of obtaining $y_{r_2}$ from $y_{r_2-1}$, we change community assignments on disjoint sets of vertices. Hence the order of these steps of changing community assignments along cycles does not affect the final community assignment mapping we obtain.  From (\ref{dg2}) we have
\begin{eqnarray}
\Gamma_2\leq \prod_{l=2}^k \left(\sum_{m_l=0}^{\infty} (nk)^{m_l\ell} e^{-\frac{(1+o(1))\Delta l m_l}{8}}\right)-1.\label{iup}
\end{eqnarray}
On the right hand side of (\ref{iup}), when expanding the product, each summand has the form
\begin{eqnarray*}
\left[(nk)^{2m_2}e^{-\frac{(1+o(1))\Delta2 m_2}{8}}\right]\cdot\left[ (nk)^{3m_3}e^{-\frac{(1+o(1))\Delta3 m_3}{8}}\right]\cdot\ldots\cdot\left[(nk)^{km_k} e^{-\frac{(1+o(1))\Delta k m_k}{8}}\right]
\end{eqnarray*}
where the factor $\left[(nk)^{2m_2}e^{-\frac{(1+o(1))\Delta 2 m_2}{8}}\right]$ represents that we changed along 2-cycles $m_2$ times, the factor $\left[ (nk)^{3m_3}e^{-\frac{(1+o(1))\Delta 3 m_3}{8}}\right]$ represents that we changed along 3-cycles $m_3$ times, and so on. Moreover, each time we changed along an $l$-cycle, we need to first determine the $l$ different colors involved in the $l$-cycle, and there are at most $k^l$ different $l$-cycles;  we then need to choose $l$ vertices to change colors, and there are at most $n^{l}$ choices.
It is straightforward to check that if $\sigma$ satisfies (\ref{ld2}), then
\begin{eqnarray*}
\lim_{n\rightarrow\infty}nke^{-\frac{(1+o(1))\Delta}{8}}=0.
\end{eqnarray*}

Therefore we have
\begin{eqnarray*}
\sum_{m_l=0}^{\infty} (nk)^{m_l\ell} e^{-\frac{(1+o(1))\Delta l m_l}{8}}\leq \frac{1}{1-e^{\log k+\log n-\frac{(1+o(1))\Delta}{8}}};
\end{eqnarray*}
when $n$ is sufficiently large and $\epsilon$ is sufficiently small.
Let 
\begin{eqnarray*}
\Psi:= \prod_{l=2}^k \left(\sum_{m_l=0}^{\infty} (nk)^{m_l\ell} e^{-\frac{(1+o(1))m_l l\Delta}{8}}\right).
\end{eqnarray*}
Since $\log(1+x)\leq x$ for $x\geq 0$, we have
\begin{eqnarray*}
0\leq \log \Psi&=&\sum_{l=2}^{k}\log \left(1+\sum_{m_l=1}^{\infty} (nk)^{m_l\ell} e^{-\frac{(1+o(1))\Delta l m_l}{8}}\right)\\
&\leq &\sum_{l=2}^{k}\sum_{m_l=1}^{\infty} (nk)^{m_l\ell} e^{-\frac{(1-\delta)\Delta l m_l}{8}}\\
&\leq &\sum_{l=2}^{k}\frac{\left(nke^{-\frac{(1-\delta)\Delta}{8}}\right)^{l}}{1-\left(nke^{-\frac{(1-\delta)\Delta}{8}}\right)^l}\\
&\leq &\frac{\left(nke^{-\frac{(1-\delta)\Delta}{8}}\right)^{2}}{\left[1-\left(nke^{-\frac{(1-\delta)\Delta}{8}}\right)^{2}\right]\left[1-\left(nke^{-\frac{(1-\delta)\Delta}{8}}\right)\right]}\rightarrow 0,
\end{eqnarray*}
as $n\rightarrow\infty$. Then
\begin{eqnarray}
0\leq \lim_{n\rightarrow\infty}\Gamma_2\leq \lim_{n\rightarrow\infty}e^{\log \Sigma}-1=0.\label{g2z}
\end{eqnarray}

Then the proposition follows from (\ref{g1z}) and (\ref{g2z}).
$\hfill\Box$

\section{Community Detection on Hypergraphs with Fixed Number of Vertices in Each Community}\label{hg2}

In this section, we study community detection on hypergraphs under the assumption that the number of vertices in each community is known and fixed. We shall prove a condition when exact recovery does not occur. 

Recall that $y\in \Omega_{n_1,\ldots,n_k}$ is the true community assignment mapping.

\noindent\textbf{Proof of Theorem \ref{p210}}.
When $y^{(ab)}\in \Omega_{n_1,\ldots,n_k}$ is defined by (\ref{yab1}),
\begin{eqnarray*}
&&t_{y^{(ab)}(a),y(a)}(y^{(ab)},y)-1=t_{y^{(ab)}(a),y(a)}(y,y)=t_{y(b),y(a)}(y,y)=0\\
&&t_{y(b),y(b)}(y^{(ab)},y)+1=t_{y(b),y(b)}(y,y)=n_{y(b)}\\
&&t_{y^{(ab)}(b),y(b)}(y^{(ab)},y)-1=t_{y(a),y(b)}(y^{(ab)},y)-1=t_{y(a),y(b)}(y,y)=0\\
&&t_{y(a),y(a)}(y^{(ab)},y)+1=t_{y(a),y(a)}(y,y)=n_{y(a)}.
\end{eqnarray*}
and
\begin{eqnarray*}
t_{i,j}(y^{(ab)},y)=t_{i,j}(y),\ \forall\ (i,j)\in \left([k]^2\setminus\{(y(a),y(a)),(y(a),y(b)),(y(b),y(a)),(y(b),y(b))\}\right)
\end{eqnarray*}
Note that 
\begin{eqnarray*}
1-p(\check{y};\sigma)\geq \mathrm{Pr}\left(\cup_{a,b\in[n],y(a)\neq y(b)}(f(y^{(ab)})-f(y)<0)\right),
\end{eqnarray*}
since any of the event $(f(y^{(ab)})-f(y)<0)$ implies $\check{y}\neq y$. By (\ref{fhxy}) we obtain that $f(y^{(ab)})-f(y)$ is a Gaussian random variable with mean value $\|\Phi*(\bA_{y^{(ab)}}-\bA_y)\|^2$ and variance $4\|\Phi*(\bA_{y^{(ab)}}-\bA_y)\|^2$.
So $1-p(\check{y};\sigma)$ is at least
\begin{eqnarray*}
&&\mathrm{Pr}\left(\cup_{a,b\in[n],y(a)\neq y(b)}(f(y^{(ab)})-f(y)<0)\right)\\
&\geq &\mathrm{Pr}\left(\mathrm{max}_{a,b\in[n], y(a)\neq y(b)}\frac{2\langle\mathbf{W},\Phi*(\mathbf{A}_{y^{(ab)}}-\mathbf{A}_{y}) \rangle}{\|\Phi*(\bA_{y^{(ab)}}-\bA_y)\|^2}>1\right)
\end{eqnarray*}
Let $H_1$, $H_2$ be given as in the assumptions of the proposition.
Then
\begin{eqnarray*}
1-p(\check{y};\sigma)\geq \mathrm{Pr}\left(\mathrm{max}_{a\in H_1,b\in H_2}\frac{2\langle\mathbf{W},\Phi*(\mathbf{A}_{y^{(ab)}}-\mathbf{A}_{y})\rangle}{\|\Phi*(\bA_{y^{(ab)}}-\bA_y)\|^2}>1\right)
\end{eqnarray*}
Let $(\mathcal{X},\mathcal{Y},\mathcal{Z})$ be a partition of $[n]^s$ defined by
\begin{eqnarray*}
&&\mathcal{X}=\{\alpha=(\alpha_1,\alpha_2,\ldots,\alpha_s)\in [n]^s: s\in\{s_1,s_1+1,\ldots,s_2\}, \{\alpha_1,\ldots,\alpha_s\}\cap (H_1\cup H_2)=\emptyset\}\\
&&\mathcal{Y}=\{\alpha=(\alpha_1,\alpha_2,\ldots,\alpha_s)\in [n]^s: s\in\{s_1,s_1+1,\ldots,s_2\}, |r\in [s]:\alpha_r\in (H_1\cup H_2)|=1\}\\
&&\mathcal{Z}=\{\alpha=(\alpha_1,\alpha_2,\ldots,\alpha_s)\in [n]^s: s\in\{s_1,s_1+1,\ldots,s_2\}, |r\in[s]:\alpha_r\in (H_1\cup H_2)|\geq2\}
\end{eqnarray*}
For $\eta\in\{\mathcal{X},\mathcal{Y},\mathcal{Z}\}$, define a random tensor $\mathbf{W}_{\eta}$ from the entries of $\mathbf{W}$ as follows
\begin{eqnarray*}
(\mathbf{W}_{\eta})_{(a_1,\ldots,a_s)}=\begin{cases}0&\mathrm{if}\ (a_1,\ldots,a_s)\notin \eta\\ \mathbf{W}_{(a_1,\ldots,a_s)},&\mathrm{if}\ (a_1,\ldots,a_s)\in \eta\end{cases}
\end{eqnarray*}
For each $u\in H_{1}$ and $v\in H_{2}$, let
\begin{eqnarray*}
\mathcal{X}_{uv}=\langle\mathbf{W}_{\mathcal{X}},\Phi*(\mathbf{A}_{y^{(uv)}}-\mathbf{A}_{y}) \rangle\\
\mathcal{Y}_{uv}=\langle\mathbf{W}_{\mathcal{Y}},\Phi*(\mathbf{A}_{y^{(uv)}}-\mathbf{A}_{y}) \rangle\\
\mathcal{Z}_{uv}=\langle\mathbf{W}_{\mathcal{Z}},\Phi*(\mathbf{A}_{y^{(uv)}}-\mathbf{A}_{y}) \rangle
\end{eqnarray*}

\begin{lemma}\label{l69}The followings are true:
\begin{enumerate}
\item $\mathcal{X}_{uv}=0$ for $u\in H_{1}$ and $v\in H_{2}$.
\item For each $u\in H_{1}$ and $v\in H_{2}$, the variables $\mathcal{Y}_{uv}$ and $\mathcal{Z}_{uv}$ are independent.
\item Each $\mathcal{Y}_{uv}$ can be decomposed into $Y_u+Y_v$ where $\{Y_u\}_{u\in H_{1}}\cup \{Y_v\}_{v\in H_{2}}$ is a collection of i.i.d.~Gaussian random variables.
\end{enumerate}
\end{lemma}

\begin{proof}
Note that for $J_s=(j_1,j_2,\ldots,j_s)\in[n]^s$,
\begin{small}
\begin{eqnarray}
&&(\mathbf{A}_{y^{(uv)}}-\mathbf{A}_{y})_{J_s}\label{kab1}\\
&=&\begin{cases}\phi(y^{(uv)}(j_1),y^{(uv)}(j_2),\ldots,y^{(uv)}(j_s))-\phi(y(j_1),y(j_2),\ldots,y(j_s))& \mathrm{if}\ \{a,b\}\cap\{j_1,\ldots,j_s\}\neq \emptyset \\0&\mathrm{otherwise}.\end{cases}\notag
\end{eqnarray}
\end{small}

It is straightforward to check (1). (2) holds because $\mathcal{Y}\cap\mathcal{Z}=\emptyset$.

For $g\in H_{1}\cup H_{2}$, let $\mathcal{Y}_g\subseteq \mathcal{Y}$ be defined by
\begin{eqnarray*}
\mathcal{Y}_g=\{\alpha=(\alpha_1,\alpha_2,\ldots,\alpha_s)\in \mathcal{Y}: g\in\{\alpha_1,\ldots,\alpha_s\}\}.
\end{eqnarray*}
Note that for $g_1,g_2\in H_{1}\cup H_{2}$ and $g_1\neq g_2$, $\mathcal{Y}_{g_1}\cap \mathcal{Y}_{g_2}=\emptyset$. Moreover, $\mathcal{Y}=\cup_{g\in H_{1}\cup H_{2}}\mathcal{Y}_g$. Therefore
\begin{eqnarray*}
\mathcal{Y}_{ab}=\sum_{g\in H_{1}\cup H_{2}}\langle\mathbf{W}_{\mathcal{Y}_g},\Phi*(\mathbf{A}_{y^{(ab)}}-\mathbf{A}_{y}) \rangle
\end{eqnarray*}
Note also that $\langle\mathbf{W}_{\mathcal{Y}_g},\Phi*(\mathbf{A}_{y^{(ab)}}-\mathbf{A}_{y}) \rangle=0$, if $g\notin \{a,b\}$. Hence
\begin{eqnarray*}
\mathcal{Y}_{ab}=\sum_{\alpha\in\mathcal{Y}_a\cup \mathcal{Y}_b}\frac{(\mathbf{W})_{\alpha}\cdot\{(\mathbf{A}_{y^{(ab)}}-\mathbf{A}_{y})_{\alpha}\}}{\sigma_{\alpha}}
\end{eqnarray*}

So, we can define
\begin{eqnarray*}
Y_a:=&&\sum_{\alpha\in \mathcal{Y}_a}\frac{(\mathbf{W})_{\alpha}\cdot\{(\mathbf{A}_{y^{(ab)}}-\mathbf{A}_{y})_{\alpha}\}}{\sigma_{\alpha}}
\end{eqnarray*}
By (\ref{kab1}) we obtain
\begin{eqnarray*}
Y_a&=&\sum_{\alpha\in \mathcal{Y}_a}\frac{(\mathbf{W})_{\alpha}\cdot\{(\mathbf{A}_{y^{(ab)}}-\mathbf{A}_{y})_{\alpha}\}}{\si_{\alpha}}\\
&=&\sum_{s=s_1}^{s_2}\sum_{j=1}^{s}\sum_{(i_1,\ldots,\widehat{i}_j,\ldots,i_s)\in([n]\setminus(H_{1}\cup H_2))^{s-1}}\frac{1}{\sigma_{(i_1,\ldots,i_{j-1},a,i_{j+1},\ldots,i_s)}}\\
&&\left\{(\phi(y(i_1),\ldots,y^{(ab)}(a),\ldots,y(i_s))-\phi(y(i_1),\ldots,y(a),\ldots,y(i_s)))(\bW)_{(i_1,\ldots,i_{j-1},a,i_{j+1},\ldots,i_s)}\right\}
\end{eqnarray*}
Similarly, define
\begin{eqnarray*}
Y_b:&=&\sum_{\alpha\in \mathcal{Y}_b}\frac{(\mathbf{W})_{\alpha}\cdot\{(\mathbf{A}_{y^{(ab)}}-\mathbf{A}_{y})_{\alpha}\}}{\sigma_{\alpha}}\\
&=&\sum_{s=s_1}^{s_2}\sum_{j=1}^{s}\sum_{(i_1,\ldots,\widehat{i}_j,\ldots,i_s)\in([n]\setminus(H_{1}\cup H_{2}))^{s-1}}\frac{1}{\sigma_{(i_1,\ldots,i_{j-1},b,i_{j+1},\ldots,i_s)}}\\
&&\left\{(\phi(y(i_1),\ldots,y^{(ab)}(b),\ldots,y(i_s))-\phi(y(i_1),\ldots,y(b),\ldots,y(i_s)))(\bW)_{(i_1,\ldots,i_{j-1},b,i_{j+1},\ldots,i_s)}\right\}
\end{eqnarray*}
Then $\mathcal{Y}_{ab}=Y_a+Y_b$ and $\{Y_g\}_{g\in H_{1}\cup H_{2}}$ is a collection of independent Gaussian random variables. Moreover, the variance of $Y_g$ is  
\begin{eqnarray*}
&&\sum_{s=s_1}^{s_2}\sum_{j=1}^{s}\sum_{(i_1,\ldots,\widehat{i}_j,\ldots,i_s)\in([n]\setminus (H_1\cup H_2))^{s-1}}\frac{1}{\sigma_{(i_1,\ldots,i_{j-1},g,i_{j+1},\ldots,i_s)}^2}\\
&&(\phi(y(i_1),\ldots,y(b),\ldots,y(i_s))-\phi(y(i_1),\ldots,y(a),\ldots,y(i_s)))^2\\
\end{eqnarray*}
By Assumption (6) of the proposition,this is independent of $g$.
\end{proof}

By the Lemma \ref{l69}, we obtain
\begin{eqnarray*}
\langle\mathbf{W},\Phi*(\mathbf{A}_{y^{(ab)}}-\mathbf{A}_{y}) \rangle=Y_a+Y_b+\mathcal{Z}_{ab}
\end{eqnarray*}
Moreover,
\begin{eqnarray*}
\max_{u\in H_{1},v\in H_{2}}Y_u+Y_v+\mathcal{Z}_{uv}&\geq& \max_{u\in H_{1},v\in H_{2}}(Y_u+Y_v)-\max_{u\in H_{1},v\in H_{2}}(-\mathcal{Z}_{uv})\\
&=&\max_{u\in H_{1}} Y_u+\max_{v\in H_{2}}Y_v-\max_{u\in H_{1},v\in H_{2}}(-\mathcal{Z}_{uv})
\end{eqnarray*}

By Lemma \ref{mg} we obtain that  when $\epsilon, h$ satisfy (\ref{epn}) with $N$ replaced by $h$, each one of the following two events
\begin{eqnarray*}
F_1:=\left\{\max_{u\in H_{1}}\frac{Y_u}{\|\Phi*(\bA_{y^{(uv)}}-\bA_y)\|^2}\geq (1-\epsilon)\sqrt{2\log h\cdot \min_{u\in H_{1}}\frac{\mathrm{Var}(Y_u)}{(L_{\Phi}(y^{(u,v)},y))^2}}\right\}\\
F_2:=\left\{\max_{v\in H_{2}}\frac{Y_v}{\|\Phi*(\bA_{y^{(uv)}}-\bA_y)\|^2}\geq (1-\epsilon)\sqrt{{2\log h\cdot \min_{v\in H_{2}}\frac{\mathrm{Var}(Y_v)}{(L_{\Phi}(y^{(u,v)},y))^2}}}\right\}
\end{eqnarray*}
has probability at least $1-e^{-h^{\epsilon}}$. Moreover, the event 
\begin{eqnarray*}
F_3:=\left\{\max_{u\in H_{1},v\in H_{2}}\frac{\mathcal{Z}_{uv}}{{\|\Phi*(\bA_{y^{(uv)}}-\bA_y)\|^2}}\leq (1+\epsilon)\sqrt{2\log( 2h)\cdot \max_{u\in H_{1},v\in H_{2}} \frac{\mathrm{Var}(Z_{uv})}{(L_{\Phi}(y^{(uv)},y))^2}}\right\}
\end{eqnarray*}
occurs with probability at least $1-h^{-2\epsilon}$. Then by Assumption (4) of the proposition we have
\begin{eqnarray*}
\mathrm{Var} \mathcal{Z}_{uv}
&=&\|\Phi*(\mathbf{A}_{y^{(uv)}}-\mathbf{A}_{y})\|^2-\mathrm{Var}(Y_u)-\mathrm{Var}(Y_v)\\
&=&L_{\Phi}(y^{(uv)},y)-(1+o(1))L_{\Phi}(y^{(uv)},y)\\
&=&o(1)L_{\Phi}(y^{(uv)},y).
\end{eqnarray*}

By Assumption (5) of the proposition, for any $u\in H_1$ and $v\in H_2$, we have
\begin{eqnarray*}
\mathrm{Var}(Y_u)=\mathrm{Var}(Y_v).
\end{eqnarray*}
Moreover, by Assumption (4) of the proposition,
\begin{eqnarray*}
\mathrm{Var}(Y_u)+\mathrm{Var}(Y_v)=(1+o(1))L_{\Phi}(y^{(uv)},y).
\end{eqnarray*}

Hence the probability of the event
\begin{eqnarray*}
&&F:=\left\{\max_{a\in H_1,b\in H_{2}}\frac{\langle\mathbf{W},\Phi*(\mathbf{A}_{y^{(ab)}}-\mathbf{A}_{y}) \rangle}{\|\Phi*(\bA_{y^{(ab)}}-\bA_y)\|^2}\geq \frac{(1-\epsilon)\sqrt{2\log h}}{\max_{u\in H_1,v\in H_2}L_{\Phi}(y^{(u,v)},y)} \right.\\
&&\left.\left(\sqrt{\min_{u\in H_1}\mathrm{Var}(Y_u)}+\sqrt{\min_{v\in H_2}\mathrm{Var}(Y_v)}-(1+o(1))\sqrt{\max_{u\in H_1,v\in H_2}\mathrm{Var}(Z_{uv})}\right)\right\}\\
&=&\left\{\max_{a\in H_1,b\in H_{2}}\frac{\langle\mathbf{W},\Phi*(\mathbf{A}_{y^{(ab)}}-\mathbf{A}_{y}) \rangle}{\|\Phi*(\bA_{y^{(ab)}}-\bA_y)\|^2}\geq \frac{2(1-\epsilon)\sqrt{\log h }}{\sqrt{\max_{u\in H_1,v\in H_2}L_{\Phi}(y^{(u,v)},y)}} (1+o(1))\right\}
\end{eqnarray*}
is at least 
\begin{eqnarray*}
\mathrm{Pr}(F_1\cap F_2\cap F_3)&=&1-\mathrm{Pr}((F_1)^c\cup (F_2)^c\cup (F_3)^c)\\
&\geq &1- \mathrm{Pr}((F_1)^c)-\mathrm{Pr}((F_2)^c)-\mathrm{Pr}((F_3)^c)\\
&\geq &1-2e^{-h^{\epsilon}}-h^{-2\epsilon}.
\end{eqnarray*}

When (\ref{as51}) holds, we have
\begin{eqnarray*}
&&\mathrm{Pr}\left(\mathrm{max}_{a,b\in[n],y(a)\neq y(b)}\frac{2\langle\mathbf{W},\Phi*(\mathbf{A}_{y^{(ab)}}-\mathbf{A}_{y} )\rangle}{\|\Phi*(\bA_{y^{(ab)}}-\bA_y)\|^2}>1)\right)
\geq \mathrm{Pr}(F)\rightarrow 1,
\end{eqnarray*}
as $n\rightarrow \infty$. Then the proposition follows.
$\hfill\Box$

\appendix

\section{Maximum of Gaussian Random Variables}\label{adl}

\begin{lemma}\label{mg}Let $G_1,\ldots, G_N$ be Gaussian random variables with mean $0$. Let $\epsilon\in (0,1)$. Then 
\begin{eqnarray*}
\mathrm{Pr}\left(\max_{i=1,\ldots,N}G_i>(1+\epsilon)\sqrt{2\max_{i\in[N]}\mathrm{Var}(G_i)\log N}\right)\leq N^{-\epsilon}
\end{eqnarray*}
and moreover, if $G_i$'s are independent, and $\epsilon, N$ satisfy
\begin{eqnarray}
\frac{N^{\epsilon-\epsilon^2}(1-\epsilon)\sqrt{2\log N}}{\sqrt{2\pi}(1+2(1-\epsilon)^2\log N)}>1\label{epn}
\end{eqnarray}
Then
\begin{eqnarray*}
\mathrm{Pr}\left(\max_{i=1,\ldots,N}G_i<(1-\epsilon)\sqrt{2\min_{j\in[N]}\mathrm{Var}(G_j)\log N}\right)\leq \exp(-N^{\epsilon})
\end{eqnarray*}
\end{lemma}

\begin{proof}It is known that for a Gaussian random variable $G_i$ and $x>0$,
\begin{eqnarray}
\frac{xe^{-\frac{x^2}{2}}}{\sqrt{2\pi}(1+x^2)}\leq \mathbf{Pr}\left(\frac{G_i}{\sqrt{\mathrm{Var}(G_i)}}>x\right)\leq \frac{e^{-\frac{x^2}{2}}}{x\sqrt{2\pi}}\label{gd}
\end{eqnarray}

Let $G_1,\ldots, G_N$ be $N$ Gaussian random variables. Then by (\ref{gd}) we have
\begin{eqnarray*}
&&\mathrm{Pr}\left(\max_{i\in[N]}G_i\geq (1+\epsilon)\sqrt{2\max_{i\in[N]}\mathrm{Var}(G_i)\log N}\right)\\
&\leq& \sum_{i\in[N]}\mathrm{Pr}\left(\frac{G_i}{\sqrt{\mathrm{Var}(G_i)}}\geq (1+\epsilon)\sqrt{2\log N}\right)\\
&\leq&\frac{N e^{-(1+\epsilon)^2\log N}}{2(1+\epsilon)\sqrt{\pi\log N}}\\
&\leq & N^{-\epsilon}
\end{eqnarray*}
If we further assume that $G_i$'s are independent, then
\begin{eqnarray*}
&&\mathrm{Pr}\left(\max_{i\in[N]}G_i< (1-\epsilon)\sqrt{2\min_{j\in[N]}\mathrm{Var}(G_j)\log N}\right)\\
&=&\prod_{i\in[N]}\mathrm{Pr}\left(G_i<(1-\epsilon)\sqrt{2\min_{j\in[N]}\mathrm{Var}(G_j)\log N}\right)\\
&=&\prod_{i\in[N]}\left[1-\mathrm{Pr}\left(G_i>(1-\epsilon)\sqrt{2\min_{j\in[N]}\mathrm{Var(G_j)}\log N}\right)\right]\\
&\leq &\prod_{i\in[N]}\left[1-\mathrm{Pr}\left(\frac{G_i}{\sqrt{\mathrm{Var}(G_i)}}>(1-\epsilon)\sqrt{2\log N}\right)\right]
\end{eqnarray*}
By (\ref{gd}) we obtain
\begin{eqnarray*}
\mathrm{Pr}\left(\max_{i\in[N]}G_i< (1-\epsilon)\sqrt{2\min_{j\in[N]}\mathrm{Var}(G_j)\log N}\right)&\leq& \left(1-\frac{(1-\epsilon)\sqrt{2\log N}}{\sqrt{2\pi}(1+2(1-\epsilon)^2\log N)}\frac{1}{N^{(1-\epsilon)^2}}\right)^N
\end{eqnarray*}
When (\ref{epn}) holds, we have 
\begin{eqnarray*}
\mathrm{Pr}\left(\max_{i\in[N]}G_i< (1-\epsilon)\sqrt{2\min_{j\in[N]}\mathrm{Var}(G_j)\log N}\right)\leq \left(1-\frac{1}{N^{1-\epsilon}}\right)^{N^{1-\epsilon}\cdot N^{\epsilon}}\leq  e^{-N^{\epsilon}}
\end{eqnarray*}
Then the lemma follows.
\end{proof}

\bigskip
\bigskip
\noindent{\textbf{Acknowledgements.}} ZL's research is supported by National Science Foundation grant 1608896 and Simons Foundation grant 638143. 
\bibliography{cdh}
\bibliographystyle{amsplain}

\end{document}